\documentclass[sigplan,screen,nonacm]{acmart}

\AtBeginDocument{%
  }

\usepackage{xspace}
\usepackage{mathtools}
\usepackage{tikz}
\usepackage{subcaption}
\usepackage{booktabs}
\usepackage{multirow}
\usepackage{array}
\usepackage{tabularx}
\usepackage[inline,shortlabels]{enumitem}


\usepackage[ruled,vlined,linesnumbered]{algorithm2e}

\newcommand*\blackcircled[1]{\tikz[baseline=(char.base)]{
            \node[shape=circle,fill,inner sep=1pt] (char) {\textcolor{white}{#1}};}}
\def\ie{{i.e.}}
\def\eg{{e.g.}}

\newcommand{\name}{FlashDiff\xspace}
\newcommand{\R}{\mathbb{R}}
\newcommand{\E}{\mathbb{E}}
\newcommand{\Var}{\mathrm{Var}}
\newcommand{\norm}[1]{\lVert #1 \rVert}
\newcommand{\abs}[1]{\lvert #1 \rvert}
\newcommand{\ftheta}{f_\theta}
\newcommand{\xhat}{\hat{x}}
\newcommand{\xstar}{x^{*}}
\DeclareMathOperator*{\argmax}{arg\,max}
\DeclareMathOperator*{\argmin}{arg\,min}

\newenvironment{denseitemize}{\begin{itemize}}{\end{itemize}}

\makeatletter
\@ifundefined{theorem}{\newtheorem{theorem}{Theorem}[section]}{}
\@ifundefined{proposition}{\newtheorem{proposition}[theorem]{Proposition}}{}
\@ifundefined{corollary}{\newtheorem{corollary}[theorem]{Corollary}}{}
\@ifundefined{lemma}{\newtheorem{lemma}[theorem]{Lemma}}{}
\theoremstyle{definition}
\@ifundefined{definition}{\newtheorem{definition}[theorem]{Definition}}{}
\@ifundefined{assumption}{\newtheorem{assumption}[theorem]{Assumption}}{}
\theoremstyle{remark}
\@ifundefined{remark}{\newtheorem{remark}[theorem]{Remark}}{}
\makeatother

\begin{document}

\title{FlashDiff: Efficient Regional Execution and Scheduling for Diffusion Model Serving}

\author{Yaqi Qiao$\footnotemark[1]^{1}$, Ping He\footnotemark[1]$^{2,\diamondsuit}$, 
Songrun Xie$^{3}$, Ayush Barik$^{1}$, Chengsong Zhang$^{4}$, Zhengzhong Tu$^{5}$, Fan Lai$^{1}$
\\
\normalsize\textit{$^1$University of Illinois Urbana-Champaign $\quad$ $^2$Vanderbilt University $\quad$ $^3$HKUST $\quad$ $^4$NVIDIA $\quad$ $^5$Texas A\&M University}
}




\begin{abstract}

Diffusion models have become the central backbone for modern image, video, and audio generation, but their efficient service remains a challenge. Unlike autoregressive decoding, diffusion inference repeatedly updates high-dimensional spatial or temporal latents over many denoising steps. This all-region execution pattern makes generation latency high and limits serving throughput. Existing multi-GPU parallelization methods (e.g., sequence parallelism) can reduce per-step computation, but often introduce substantial activation exchange overhead, causing communication to offset or even outweigh the benefits of parallel execution.

This paper presents \name, a diffusion serving system that improves inference efficiency through adaptive regional execution and scheduling. \name is based on the observation that diffusion refinement is not uniform across latent regions or denoising steps: different regions often stabilize at different rates, while neighboring steps exhibit strong temporal correlation. \name leverages these properties to selectively execute only regions that require further refinement and to reallocate the resulting compute slack across concurrent serving requests. \name consists of three mechanisms. First, it decomposes the latent representation into coherent execution regions using early-stage attention signals, preserving semantic structure while exposing fine-grained parallelism. Second, it uses a lightweight runtime controller to estimate region activity and bypass low-impact updates when further refinement is unlikely to affect output quality. Third, it applies an affinity-aware online scheduler that co-locates dependent regions, balances residual load across GPUs, and reuses reclaimed compute capacity to improve serving efficiency. 
Across real-world image, video, and audio workloads, \name reduces end-to-end serving latency by 30–97\% and improves throughput by 1.2–2.2×, by eliminating 24–66\% of  computation without degrading  quality.

\end{abstract}


\maketitle


\footnotetext{* Equal contribution. \(\diamondsuit\) Work was done when visiting at UIUC.}

\section{Introduction}
\label{sec:intro}

Diffusion models underpin modern text-to-image (T2I)~\cite{sd15, sdxl}, text-to-video (T2V)~\cite{sora, hunyuanvid}, and text-to-audio (T2A) generation~\cite{stableaudioopen, audioldm, tango}, powering large-scale services such as Adobe Firefly, which produces millions of images per day~\cite{niavana-nsdi24}. Unlike the token-by-token decoding of large language models (LLMs), diffusion models operate over high-dimensional spatial and temporal latents through tens of steps (e.g., often 50), where each step predicts noise and progressively refines the previous latent toward a coherent output~\cite{sd3, diffusion-mlsys-sigmetrics25}. Despite their moderate parameter size, typically a few billion parameters~\cite{wan2025}, diffusion models impose substantial inference cost. Each denoising step updates the entire latent representation (e.g., spatial feature maps in T2I or temporal sequences in T2A), making execution fundamentally compute-bound. Even with recent advances in reducing denoising steps~\cite{dpm-solver, seeds-nips23, sdstep-arxiv23} and simplifying architectures~\cite{snapfusion-nips23}, generating moderate-resolution outputs can take a dozen seconds on H100 GPUs~\cite{distrifusion-cvpr24}. Unfortunately, with existing multi-GPU parallelism strategies, such as tensor and sequence parallelism~\cite{alpaserve-osdi23, fang2024xdit, asyndiff-nips24}, the communication payloads (e.g., for exchanging activations) often outweighs the compute savings these strategies aim to provide (\S\ref{subsec:motivation}).
This has led to severe latency and throughput bottlenecks in production~\cite{niavana-nsdi24}.

This paper introduces \name, a diffusion serving engine that breaks the rigid, monolithic execution of denoising into semantic-aware, parallelizable subtasks. Although latent elements are globally coupled at each step, \name exploits two intrinsic properties of diffusion: (1) \emph{Spatial heterogeneity}: different latent regions refine at different rates (e.g., smooth backgrounds versus detailed objects), allowing low-complexity regions to skip denoising steps without perceptual degradation; and (2) \emph{Temporal affinity}: latent states across adjacent denoising steps are  correlated~\cite{distrifusion-cvpr24,asyndiff-nips24}, enabling reuse of prior states when updates are skipped. These properties motivate a new execution principle, \emph{semantic patch parallelism}, which decouples \emph{where} computation is needed from \emph{when} it must be performed. The key idea is to partition the latent into semantically coherent patches---contiguous spatial regions (for images and video) or temporal segments (for audio)---and selectively refine each patch across denoising steps. This enables fine-grained, complexity-aware execution that fundamentally differs from, and is not achievable with existing model-parallel approaches.

Realizing semantic patch parallelism requires overcoming three fundamental efficiency-quality challenges (\S\ref{subsec:motivation}):  
(1) \emph{Partitioning tax}: object-level partitioning over-fragments the latent space, increasing synchronization overhead, while uniform partitioning ignores semantic boundaries, leading to incoherent patches and visual artifacts; 
(2) \emph{Error propagation}: skipping steps reduces computation but risks local errors propagating to neighboring regions, especially as refinement needs evolve with the changing alignment between latent features and the input prompt across steps; and  
(3) \emph{Execution bubbles}: semantic partitioning may produce patches with heterogeneous sizes, creating stragglers and complicating concurrent request execution under dynamic serving loads.

\name makes semantic patch parallelism practical by jointly addressing where to compute, when to compute, and how to schedule computation at scale, through three architectural components: the \emph{Patch Partitioner}, \emph{Patch Gate}, and \emph{Patch Scheduler} (\S\ref{sec:overview}). 
The partitioner identifies \emph{where} computation is needed by estimating the refinement demand of latent regions using cross-attention signals from early denoising steps. Semantically coherent latent regions with similar refinement velocities are grouped into patches to preserve structure while enabling parallel execution. Crucially, it dynamically adapts patch granularity based on runtime communication-to-computation ratio,  serving load, and the distribution of region complexity (\S\ref{subsec:partitioner}).

The gate provides runtime control over \emph{when} each patch should be refined. It captures spatiotemporal changes in self-attention to estimate both intra-patch activity and cross-patch influence. When a patch's dynamics diminish, the current step is skipped; when dependencies re-emerge, the patch is reactivated. This design exploits temporal affinity while preserving global coherence (\S\ref{subsec:gate}). The scheduler then translates this harvested semantic slack into system-level gains. It introduces lightweight, step-level patch switching to minimize overhead for accommodating high-need patches and to redistribute freed compute across requests to improve the global efficiency-quality frontier, such as maximizing global generation quality under low serving loads. To mitigate communication overhead and stragglers, it applies affinity-aware packing, co-locating dependent patches and balancing residual load across workers (\S\ref{subsec:scheduler}).

\name generalizes across model architectures and modalities, supporting image, video, and audio generation. We implement \name atop NVIDIA TensorRT~\cite{trt-llm}, ensuring seamless integration with existing diffusion serving stacks (\S\ref{sec:implementation}). Comparing to state-of-the-art engines including xDiT~\cite{fang2024xdit} and DistriFusion~\cite{distrifusion-cvpr24}, across real-world T2I, T2V, and T2A workloads (e.g., SD3~\cite{sd3}, FLUX~\cite{flux2024}, WAN2~\cite{wan2025}), \name reduces end-to-end latency by 30--97\% and improves throughput by 1.2--2.2$\times$, by eliminating 24--66\% of computation without compromising generation quality (\S\ref{sec:eval}).

Overall, this paper makes the following contributions:
\vspace{-4pt}

\begin{denseitemize}
\item We propose semantic patch parallelism, a novel execution paradigm that treats the latent grid as a collection of heterogeneous, semantically-governed compute tasks, decoupling \emph{where} and \emph{when} refinement occurs to reduce generation execution (\S\ref{sec:background}--\S\ref{sec:overview}).
\item We develop new mechanisms to perform efficiency- and quality-aware patch partitioning, patch-level execution gating, and affinity-aware scheduling (\S\ref{sec:design}).
\item We show \name's substantial efficiency gains on real-world T2I, T2V, and T2A workloads  (\S\ref{sec:implementation}--\S\ref{sec:eval}).
\end{denseitemize}

\section{Background and Motivation}
\label{sec:background}

\begin{figure}[t]
  \centering
  \includegraphics[width=0.99\linewidth]{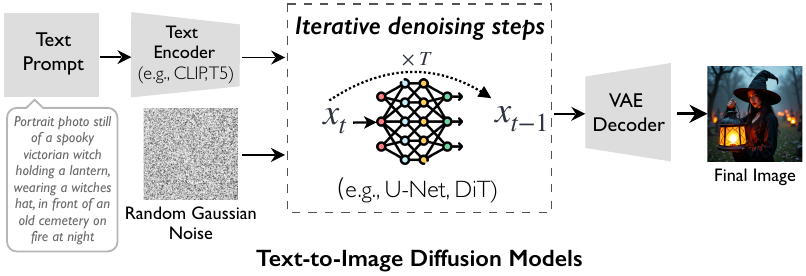}
  \vspace{-.25cm}
  \caption{Diffusion models generate images, videos, and audio through many iterative steps. Each step refines all latent elements, making inference compute-intensive. 
}
  \label{fig:dm-arc}
\vspace{-.3cm}
\end{figure}

\subsection{Diffusion Model Serving}
\label{subsec:diff-serving}

Diffusion models generate high-fidelity outputs via a reverse diffusion process, which gradually denoises an initial Gaussian latent tensor $\mathbf{x}_T$ into a coherent latent representation $\hat{\mathbf{x}}_0$ over many denoising steps (often 20--100~\cite{distrifusion-cvpr24, modm-asplos26, niavana-nsdi24}). 
Figure~\ref{fig:dm-arc} illustrates this process for text-to-image (T2I) generation; T2V and T2A generation follow a similar iterative workflow yet differ in their encoders and decoders. 
Modern diffusion models are increasingly based on the Diffusion Transformer (DiT) architecture due to its superior performance and compatibility with different modalities~\cite{dit}. DiTs employ two complementary forms of attention: (1) \emph{Cross-attention}, which conditions each latent token (spatial/temporal latent position) on the prompt text-token, and (2) \emph{Self-attention}, which models interactions among latent tokens to capture long-range structure and global coherence.

Practical deployments must satisfy two pressing needs: low latency for interactive user experiences, and high throughput to sustain large request volumes. Adobe’s Firefly integrates T2I generation into tools such as Photoshop and Express, where preview results must be produced in seconds to preserve creative workflow fluidity~\cite{adobe2025fireflymisc, niavana-nsdi24}. 
Generative advertising platforms like Google Performance Max require rendering in seconds to avoid measurable revenue losses~\cite{naresh2023performance-max}. 
Suno reports over 10 million T2A generation   users and  highlights a similar need for interactive generation~\cite{suno2024blog}.

\subsection{Motivations for Semantic Patch Parallelism}
\label{subsec:motivation}

Facing stringent latency and throughput requirements, scaling model execution across GPUs has historically delivered substantial success for LLMs. Existing multi-GPU techniques, such as tensor and sequence parallelism~\cite{alpaserve-osdi23, Megatron-lm}, often benefit from higher compute-to-memory arithmetic intensity by partitioning model weights or contexts to accommodate ever-larger models and longer sequences. However, we identify a fundamental mismatch between these paradigms and the dynamic semantic nature of diffusion workloads.

\paragraph{Existing parallelism falters due to the diffusion communication wall.}
Unlike LLMs, diffusion models typically contain only a few billion parameters, yet each denoising step is extremely compute-intensive. As shown in Figure~\ref{fig:gpu_utilization_bars}, generating a small  512$\times$512 image with a state-of-the-art model such as Flux~\cite{flux2024} drives an H100 GPU to 98.9\% utilization. This per-step saturation directly translates into high end-to-end latency and severely caps serving throughput.

Although multi-GPU execution could in principle mitigate this compute bottleneck, diffusion models possess a uniquely unfavorable activation-to-weight ratio: each denoising step requires exchanging large intermediate activations and 1D–3D latent feature maps. Worse, these transfers scale with resolution and GPU count. As shown in Figure~\ref{fig:model-parallelism}, even with diffusion-optimized sequence parallelism (xDiT~\cite{fang2024xdit}), communication still dominates per-step latency at 8 GPUs (e.g., 76\% in the SD3 model). Consequently, parallel efficiency collapses and even yields negative returns.

\begin{figure}[t]
  \centering
  \begin{minipage}[t]{0.48\linewidth}
    \centering
    \includegraphics[width=\linewidth]{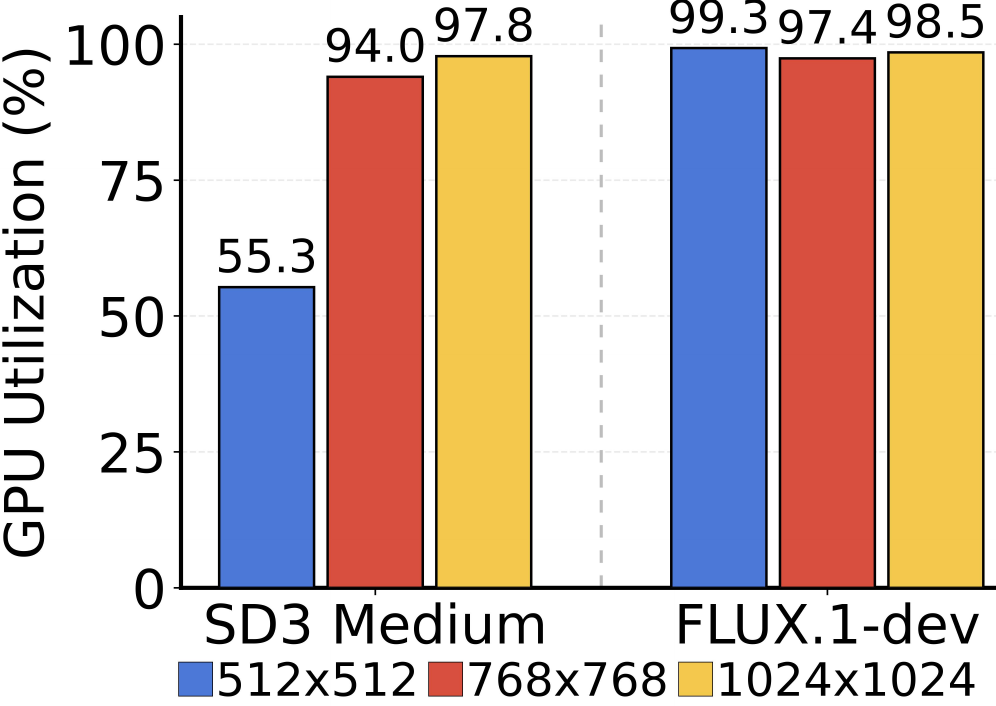}
    \caption{Diffusion is compute-bound, saturating GPUs at a small latent size. }
    \label{fig:gpu_utilization_bars}
  \end{minipage}
  \hfill
  \begin{minipage}[t]{0.48\linewidth}
    \centering
    \includegraphics[width=\linewidth]{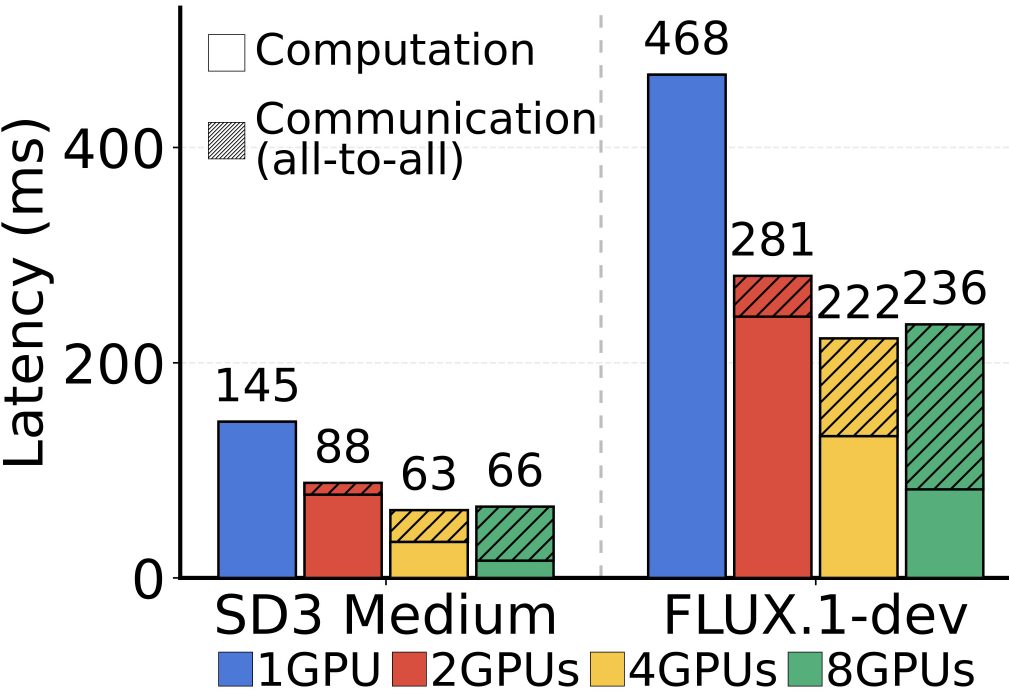}
    \caption{Communication overheads often overshadow multi-GPU gains. }
    \label{fig:model-parallelism}
  \end{minipage}
\vspace{-.5cm}
\end{figure}

\paragraph{Semantic heterogeneity within monolithic diffusion.}
Diffusion models operate over \emph{spatial or temporal latent fields} (e.g., image regions, video segments, or audio windows) whose local refinement difficulty varies widely. This heterogeneity is intrinsic to the underlying content. For example, in the prompt ``a white long-haired cat sitting before a black wall,’’ the detailed fur and contours of the cat require more iterative correction than the nearly uniform background. To quantify this generality, we study T2I generation using the FLUX model~\cite{flux2024}, T2V using the WAN-2.1 model~\cite{wan2025}, and T2A using the StableAudioOpen model~\cite{stableaudioopen}, each evaluated on thousands of real user prompts (e.g., DiffusionDB for T2I~\cite{wangDiffusionDBLargescalePrompt2022}) with a standard 50-step denoising schedule. We perform post-hoc analysis to determine how many denoising steps can be skipped for each region without degrading output quality (detailed experiment setups in Section~\ref{eval:setup}). As shown in Figure~\ref{fig:patch_complexity_magnitude}, regions exhibit sharply different refinement rates: approximately 23\% converge within 30 steps, while others require nearly the full denoising steps.

However, existing serving engines view denoising as a \emph{monolithic} iterative process, mandating equal execution for every latent region regardless of its refinement needs. This rigid execution model fails to exploit semantic heterogeneity, where latent regions refine at heterogeneous rates.

\paragraph{Systemic challenges in harvesting semantic heterogeneity.}
Translating this semantic heterogeneity into systems efficiency introduces a fundamental efficiency-quality tension: selective refinement of latent regions could reduce the amount of execution needed per request, but skipping refining regions or denoising steps can degrade generation quality. This raises system-architectural challenges:
\begin{denseitemize}
\item \emph{Runtime dynamics}: Refinement complexity is not static; it depends on the evolving non-linear alignment between the text prompt and the latent state. A system must identify skippable regions at runtime with minimal overhead, in a few milliseconds.

\item \emph{Compute-communication tradeoffs}: Partitioning the latent into semantic patches complicates data dependencies. If patches are too small or semantically fragmented, the synchronization tax to maintain global coherence can eclipse the compute savings (Figure~\ref{fig:model-parallelism}).

\item \emph{Execution bubbles}: Going beyond single requests, practical serving must handle high request concurrency. As shown in Figure~\ref{fig:patch-area}, semantic regions vary significantly in size and complexity. Region-level partitioning can lead to uneven patch sizes, thus imbalanced loads that complicate scheduling, making it difficult to sustain both high throughput and low latency.

\end{denseitemize}

These challenges necessitate a new execution paradigm that treats latent regions as first-class, schedulable compute units rather than monolithic tensors.

\begin{figure}[t]
  \centering
  \begin{minipage}[t]{0.48\linewidth}
    \centering
    \includegraphics[width=\linewidth]{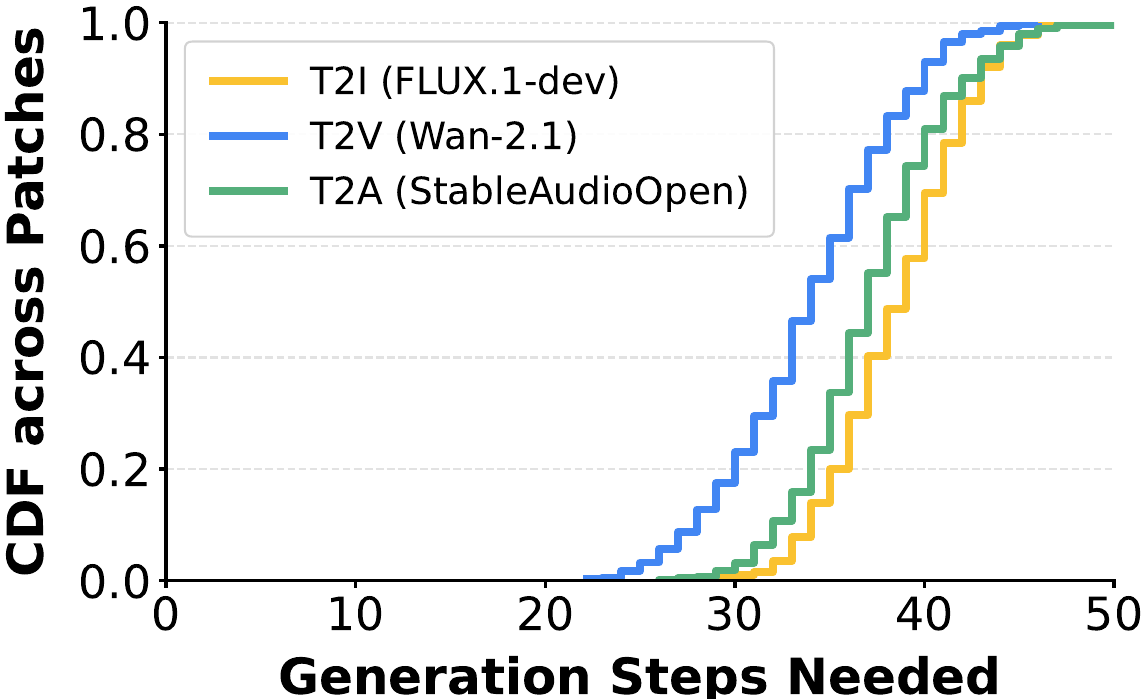}
    \caption{Latent regions exhibit heterogeneous refinement needs. }
    \label{fig:patch_complexity_magnitude}
  \end{minipage}
  \hfill
  \begin{minipage}[t]{0.48\linewidth}
    \centering
    \includegraphics[width=\linewidth]{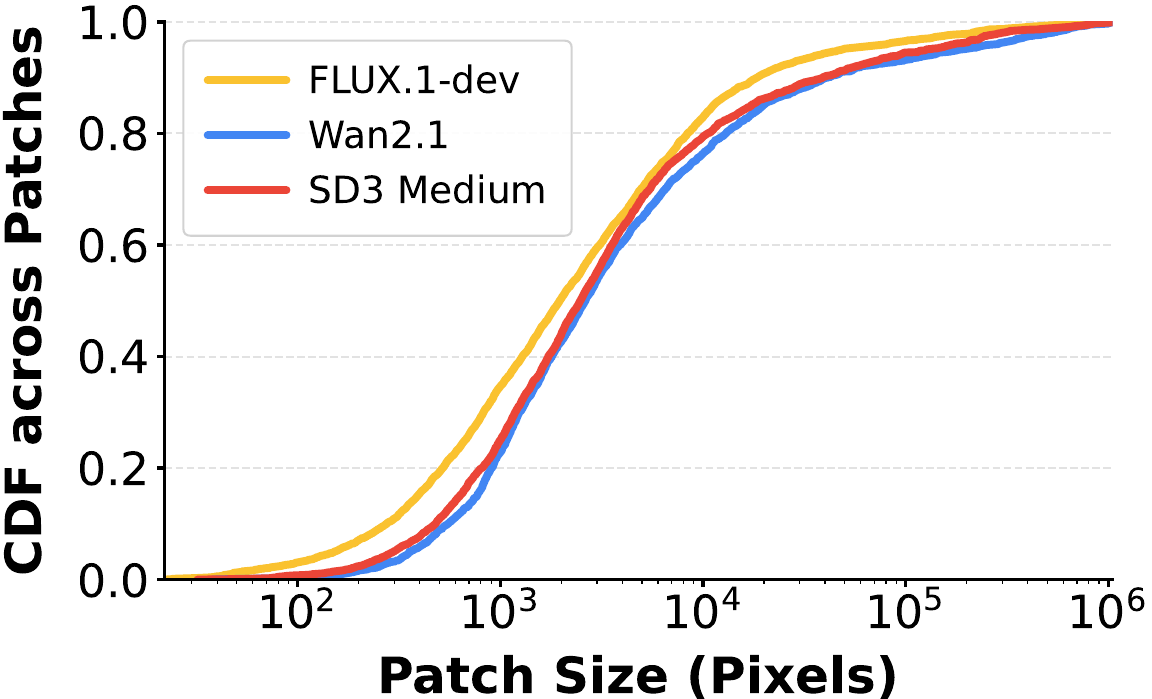}
    \caption{Region sizes vary widely, producing imbalanced patch workloads. }
    \label{fig:patch-area}
  \end{minipage}
\vspace{-.4cm}
\end{figure}

\section{\name Overview}
\label{sec:overview}

We introduce \name, a serving engine that transforms the rigid, monolithic diffusion process into adaptive, semantic-aware execution, while supporting T2I, T2V, and T2A generation. For a given request (e.g., from the upstream scheduler), it orchestrates execution across workers to jointly optimize generation latency and system throughput. 

The core of \name is the abstraction of semantic patch parallelism that partitions the latent grid into \emph{patches}---semantically coherent ``compute units'' that represent distinct spatial or temporal regions. This abstraction allows the system to treat generative refinement as a scheduling problem, decoupling the refinement need from resource allocation.

\begin{figure}[t]
  \centering
  \includegraphics[width=\linewidth]{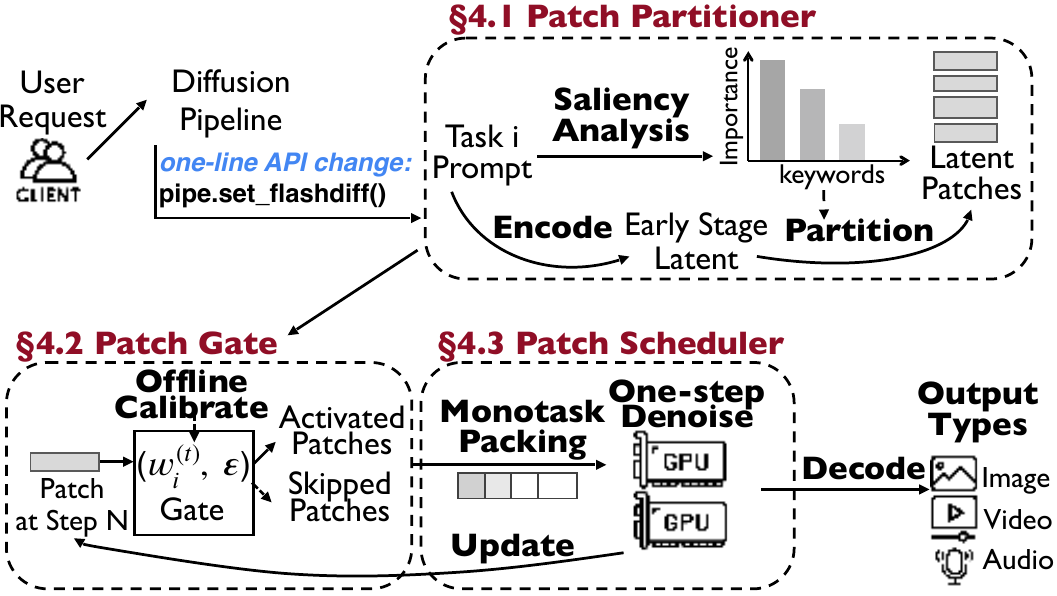}
  \caption{\name overview with three key components.
  }
  \label{fig:sys-overview}
  \vspace{-.3cm}
\end{figure}

\paragraph{Workflow.}
As shown in Figure~\ref{fig:sys-overview}, \name operates in an online serving setting where requests arrive dynamically, integrating seamlessly with existing diffusion-model serving stacks through only a few-line API change. 
Upon receiving a request from the upstream scheduler~\cite{niavana-nsdi24, tetriserve-arxiv25}, \blackcircled{1} \emph{Patch partitioning}: 
\name performs a brief warm-up phase identical to standard diffusion for the first few denoising steps. During this phase, it extracts early semantic signals (e.g., cross-attention maps) and invokes the \emph{Patch Partitioner} to group latent elements into semantically coherent patches. The partitioner balances semantic coherence, refinement complexity, and system load when determining patch granularity.
\blackcircled{2} \emph{Adaptive execution}: 
During subsequent denoising steps, the \emph{Patch Gate} determines whether each patch should be refined. It tracks spatiotemporal changes in attention to estimate refinement importance and selectively skips patches with low activity, while allowing them to resume when dependencies re-emerge. Active patches exchange updated latent representations at each denoising step, while skipped patches reuse cached states.
\blackcircled{3} \emph{Patch scheduling}:
The \emph{Patch Scheduler} assigns active patches to GPUs and reclaims resources freed by skipped patches. It adaptively redistributes computation across requests to load balance and improve overall system efficiency. 

\section{\name Design}
\label{sec:design}

We next describe how \name performs quality- and load-aware partitioning to generate semantically coherent patches (\S\ref{subsec:partitioner}), selectively skips patch execution over denoising steps (\S\ref{subsec:gate}), and schedules patches across workers to minimize serving latency and maximize throughput (\S\ref{subsec:scheduler}).

\subsection{Patch Partitioner: Enable Semantic Patch Parallelism}
\label{subsec:partitioner}

Unlike unstructured sparsity (e.g., skipping individual pixels), which current GPUs cannot efficiently exploit due to their reliance on contiguous tensor kernels, we target \emph{patch-level granularity} to expose structured sparsity that aligns with modern accelerators and distributed execution primitives.

An effective patch partitioning strategy must balance (i) \emph{semantic coherence} to maintain the structural integrity of the latent grid, (ii) \emph{complexity affinity} to group patches with similar refinement trajectories, allowing for collective execution gating, and (iii) \emph{systems efficiency}, ensuring that the partitioning logic itself does not introduce a prohibitive runtime tax or execution bubbles.

However, meeting these needs reveals multifold challenges. Fine-grained partitioning (numerous micro-patches) maximizes the theoretical parallelism and skipping potential but imposes a substantial cross-patch synchronization tax. Yet, rigid, coarse-grained partitioning (\eg, uniform grids) suffers from semantic entanglement, where a single high-complexity object forces its low-complexity neighbors to remain in the compute-intensive path. 
Furthermore, off-the-shelf segmentation primitives like Segment Anything~\cite{sam-iccv23} are computationally heavy, produce fragmented outputs, rely on high-fidelity features unavailable during noise-dominated early steps, and are often constrained to vision tasks.

\begin{figure}[t]
    \centering
  \includegraphics[width=\linewidth]  {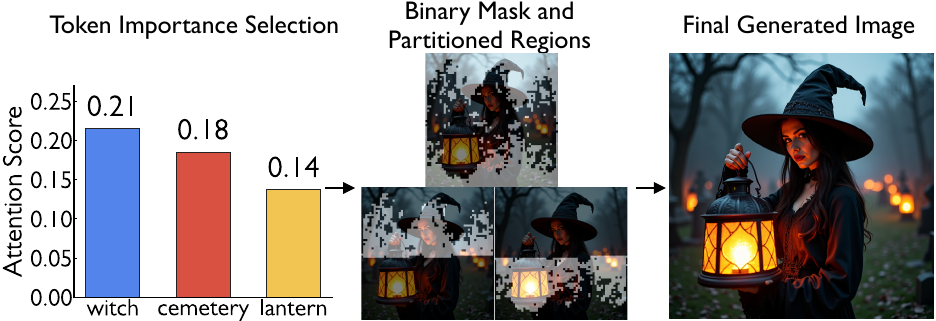}
    \caption{Semantic-aware latent partitioning workflow.}
    \label{fig:workflow_real_image2}
    \vspace{-.3cm}
\end{figure}

\paragraph{Capturing Patch Coherence and Complexity.}
\name leverages the model's cross-attention maps to jointly capture semantic coherence and refinement need, enabling model-aware partitioning too. At each denoising step, diffusion models already compute the attention alignment between latent positions (\eg, spatial for images or temporal in audio) and prompt tokens. Existing ML theory~\cite{liu2024understanding, SelfAttnWhy2} shows that these cross-attention patterns stabilize early, while later steps primarily refine local details (e.g., texture) rather than introduce new regions. Therefore, \name extracts cross-attention maps after a short warmup phase (e.g., the first five steps) to guide partitioning. We focus on salient tokens (\eg, nouns, verbs, and descriptive adjectives) that correspond to concrete semantic entities, while excluding function words (e.g., "the", "on"). These salient tokens can be identified via a part-of-speech tagger~\cite{pos-tag} applied directly to the prompt. This design generalizes well across modalities and remains lightweight (\S\ref{eval:e2e}). 
As shown in Figure~\ref{fig:workflow_real_image2}, for a set of salient token indices $\mathcal{T}$, we aggregate their attention maps to yield a semantic saliency map $S$:
$
S_j = \frac{1}{|\mathcal{T}|} \sum_{t \in \mathcal{T}} A_{j,i},
$
where $A_{j,i}$ is the attention weight from latent location $j$ to token $i$. This map serves as a compute density signal: focus regions with higher $S_j$ represent high-entropy semantic units that require sustained refinement to avoid perceptual degradation. 

\paragraph{Recursive Load-Aware Patch Partitioning.}
Given the saliency map, \name partitions the latent into patches that expose parallelism while maintaining load balance. Standard partitioning methods (e.g., K-means clustering) are ill-suited, as they ignore spatial or temporal locality and produce fragmented, non-contiguous regions. Instead, we introduce a recursive saliency partitioning mechanism. Inspired by Otsu's method~\cite{otsu1975threshold}, which thresholds grayscale images by maximizing the variance between foreground and background pixels, our partitioner bisects the given continuous saliency map $\tilde{S}_j$ by thresholding the saliency map: $M_j = \mathbf{1}[\,\tilde{S}_j \ge \theta\,]$. Here, $\theta$ is automatically chosen to maximize inter-region variance. This produces two contiguous latent token categories: a \emph{focus} region (high saliency) and a \emph{context} region (low saliency).

To further ensure load balance, given the target number of partitions (regions) $R$, we partition the current latent $r$ based on the relative token density. Let $n_f=\sum_j M_j$ and $n_c=J-n_f$ denote the number of focus and context tokens, where $J = H_{\text{lat}} \times W_{\text{lat}}$ is the total number of tokens in the latent space. For example, $H_{\text{lat}}$ and $W_{\text{lat}}$ correspond to the height and width of the image latent grid, respectively. We allocate $r_f$ focus regions by
$
r_f = \operatorname*{argmin}_{r \in \{1,\dots,R-1\}}
\left| r - R \frac{n_f}{n_f + n_c} \right|$ to balance the token complexity. 
We then split the focus and context token index lists into $r_f$ and $r_c=R - r_f$. 
This process continues until $R$ patches are obtained. 
The result is a set of semantically coherent,  contiguous patches with balanced workloads.

Figure~\ref{fig:region_psnr} shows that our partitioning consistently outperforms uniform partitioning in final generation quality across a range of partition counts (normalized to the latter; experiment setups in \S\ref{eval:setup}). 
Furthermore, Figure~\ref{fig:step-to-quality} demonstrates that, under equivalent execution budgets (e.g., \name skips 20 out of 50 steps versus directly running 30 steps without \name), \name achieves higher generation quality. This highlights that selectively allocating computation to semantically important regions is fundamentally more effective than uniformly reducing denoising steps.

\begin{figure}[t]
    \centering
    \begin{minipage}[t]{0.48\linewidth}
        \vspace{0pt}
        \centering
        \includegraphics[width=\linewidth]
        {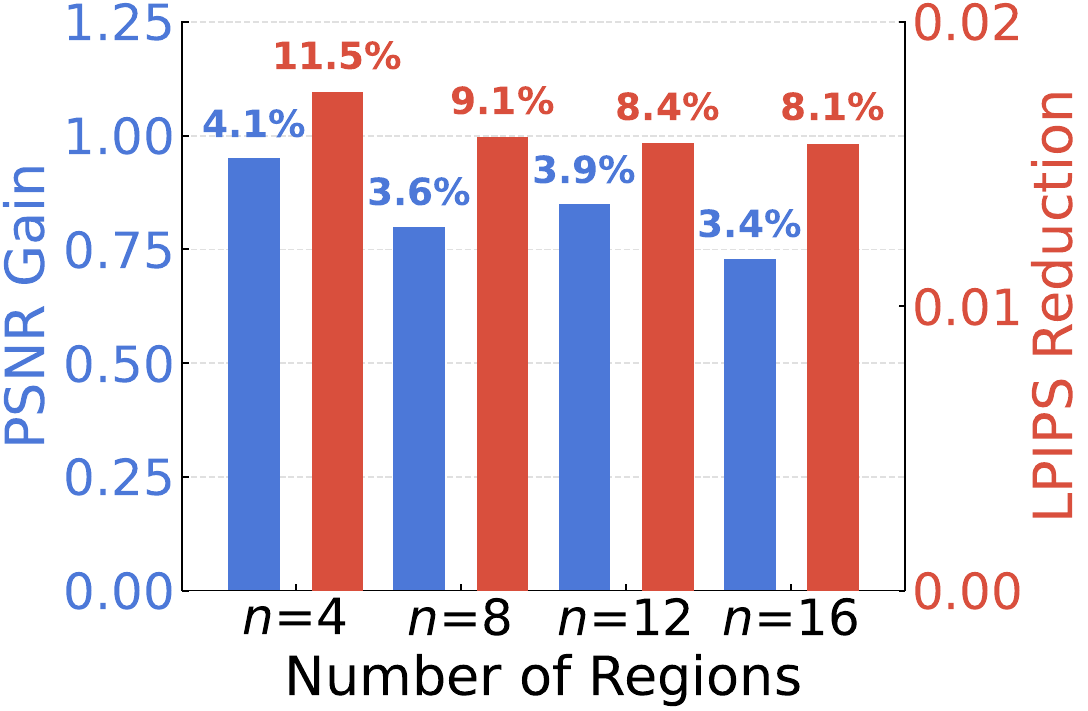}
        \captionof{figure}{Our partitioning outperforms the uniform partitioning method. 
        }
        \label{fig:region_psnr}
    \end{minipage}
    \hfill
    \begin{minipage}[t]{0.485\linewidth}
        \vspace{0pt}
        \centering
        \includegraphics[width=\linewidth]
    {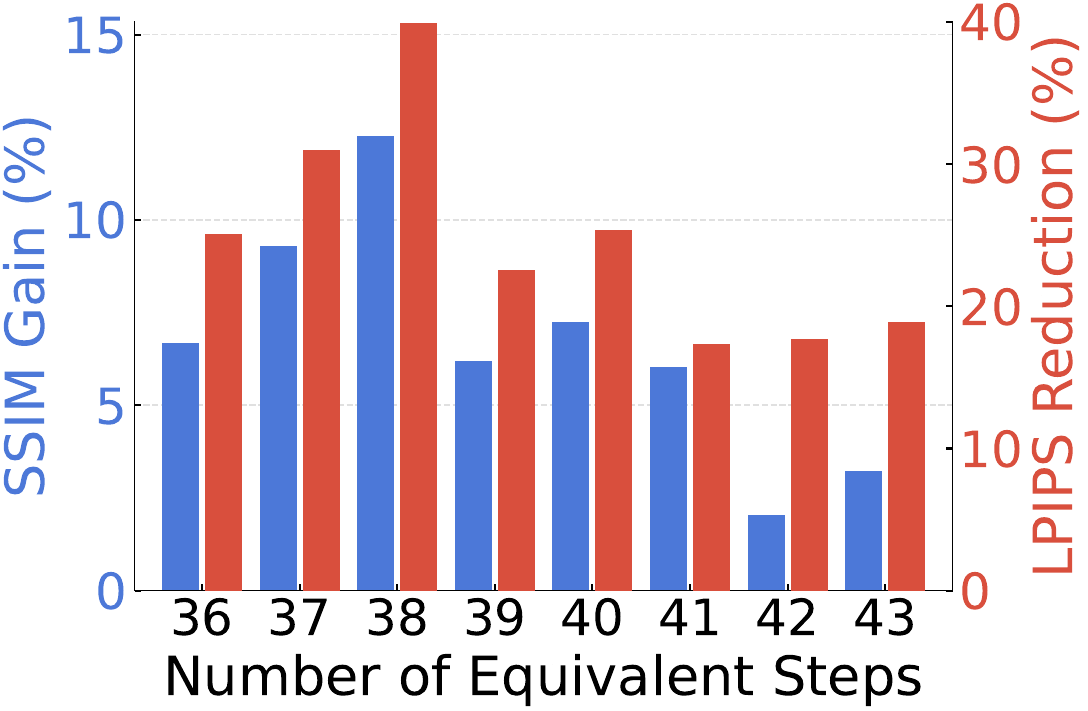}
        \captionof{figure}{Our skipping outperforms uniform step reduction under equivalent effective compute budgets.
        }
        \label{fig:step-to-quality}
    \end{minipage}
    \vspace{-0.3cm}
\end{figure}

\paragraph{Adaptive Patch Granularity.}
A remaining challenge is determining the number of patches \(R\). Increasing \(R\) exposes more fine-grained skipping opportunities, but also amplifies cross-patch communication overhead and risks fragmenting semantically coherent regions. Determining the optimal \(R\) a priori is impractical, as it depends on the input prompt, model dynamics, and future skip opportunities.

Instead, \name leverages a key structural property of this tradeoff. As \(R\) increases, the opportunity for selective refinement grows monotonically, improving theoretical computational efficiency. However, system-level costs, particularly cross-patch communication, grow rapidly and eventually dominate execution time. At the same time, generation quality remains stable under moderate partitioning but degrades when patches become overly fragmented (\S\ref{eval:ablation}). 
Together, these effects induce a \emph{unimodal} performance profile with respect to \(R\): performance improves with increasing parallelism up to a point, after which communication overhead and quality degradation outweigh the benefits. This structure enables efficient online optimization for achieving the efficiency-quality frontier. Specifically, \name performs a bounded binary search over \(R\), guided by observed end-to-end latency and quality signals from past requests. When increasing \(R\) improves performance, the system continues exploring finer partitioning; when performance degrades, it reduces \(R\) to avoid excessive communication and fragmentation. This self-adaptation can quickly and continuously converge to the near-optimal operating point under the current workload.
Our online serving deployments confirm the effectiveness of this design (\S\ref{eval:e2e}).

\subsection{Patch Gate: Selectively Refine Patches}
\label{subsec:gate}

Given the partitioned patches, the Patch Gate determines \emph{when} each patch should be refined along the denoising trajectory. This is non-trivial because refinement needs evolve dynamically: complex patches may require substantial updates early but stabilize later, while low-saliency patches may still need updates to preserve global coherence for other patches due to cross-patch dependencies.

\paragraph{Adaptive Skipping under Spatiotemporal dynamics.}
We use self-attention as a proxy for refinement activity. Unlike cross-attention (used by the Patch Partitioner) that captures the alignment of the latent region with the prompt, self-attention reflects how the latent locations interact with each other and how updates propagate across regions~\cite{liu2024understanding, hong2023sag}. 
Moreover, due to the \emph{bidirectional} attention in diffusion models, each latent position both attends to and is influenced by others. So patches with strong self-attention activity are important for their own refinement and guiding other patches.

We quantify this using a refinement importance score $R_i^{(t)}$ for each patch $i$ at step $t$. Let $\mathcal{A}_i$ denote the set of latent locations belonging to patch $i$, and let $SA \in \mathbb{R}^{J \times J}$ be the self-attention matrix with $J$ total tokens. We define:
$
R_i^{(t)} = \frac{1}{|\mathcal{A}_i|^2} \sum_{u, v \in \mathcal{A}_i} SA^{(t)}_{u,v}
$, which captures both internal and externally dependent attention strength of patch $i$.

Intuitively, patches with high refinement importance $R_i$ should not be skipped. However, $R_i$ varies both across patches and over time, making static thresholding unreliable. Moreover, we would hope to incorporate the relative efficiency gains from skipping different-sized patches. Fortunately, our recursive partitioning strategy (\S\ref{subsec:partitioner}) has produced patches of approximately equal area for load balance, even though their shapes (e.g., width and height) may differ, thereby sidestepping reasoning about heterogeneous gains.

Our key idea is to gate patches based on how much
their refinement signal is \emph{changing} relative to other patches, rather than on
its absolute magnitude. For each patch $i$, we maintain (i) $\bar R_i$, the most recent importance when the patch was executed, and (ii) $\Delta_i$, a cached change magnitude.
When executed at step $t$, we update $
\Delta_i^{(t)}=\big|R_i^{(t)}-\bar R_i\big|$ and $\bar R_i \leftarrow R_i^{(t)}$. 
When the patch is skipped, we retain $\Delta_i^{(t)}=\Delta_i^{(t-1)}$ and then normalize these changes across patches: $w_i^{(t)}=\frac{\Delta_i^{(t)}}{\sum_j \Delta_j^{(t)}}$ so that $w_i^{(t)}$ measures how much patch $i$ contributes to the \emph{overall
refinement activity} of the latent. 
A patch is skipped if $w_i^{(t)} < \varepsilon$. This decision is both step-aware and prompt-adaptive: patches are skipped only when their refinement dynamics are negligible relative to others. We show that \name achieves consistently better performance than existing advances across a wide span of $\varepsilon$ (\S\ref{eval:ablation}).

\paragraph{Guaranteeing Generation Quality.} When a patch is skipped, its noise prediction is reused from the most recent active step, incurring a local approximation error. Yet, we prove that our design introduces a \emph{bounded and well-behaved error} across the denoising trajectory, ensuring good quality:  
\begin{theorem}[Quality Bound]
Let $\hat{x}_N$ denote the final latent produced by \name after $N$ denoising steps, and let $x^*_N$ denote the baseline latent without gating. Under the adaptive gating rule with threshold $\varepsilon$, the deviation satisfies
$$
\|\hat{x}_N - x^*_N\| \;=\; O\!\Big(\varepsilon \cdot \textstyle\sum_{n=1}^N D_n \;+\; \gamma\Big),
$$
where $D_n$ is the total refinement activity at step $n$, and $\gamma$ is a bounded residual induced by forced reactivation. In particular, the quality gap vanishes as $\varepsilon \to 0$.
\end{theorem}

\emph{Proof sketch.}
The result follows from three key observations. 
(1) \emph{Bounded local error:} A patch is skipped at step $n$ only if its normalized refinement contribution satisfies $w_i^{(n)} < \varepsilon$, implying its change magnitude $\Delta_i^{(n)} \le \varepsilon \cdot D_n$. 
(2) \emph{Lipschitz continuity:} The post-attention components (MLP, normalization, projection) are Lipschitz-continuous, so perturbations in attention induce proportionally bounded deviations in the predicted noise. Denoting the per-step prediction error by $\delta_n^{(i)}$, we obtain $\delta_n^{(i)} = O(\varepsilon \cdot D_n)$. 
(3) \emph{Bounded staleness:} A forced-reactivation constraint ensures that no patch is skipped for more than a constant number of steps (\S\ref{sec:implementation}), bounding the drift from reused predictions by a small additive term $\gamma$. 
Aggregating these per-step errors across the denoising trajectory and applying a discrete Gronwall inequality~\cite{gronwall-equality} yields the stated bound. A tighter analysis exploiting contractivity of the reverse diffusion process is deferred to Appendix~\ref{sec:theory}. \hfill $\square$

Empirically, our evaluations across diverse T2I, T2V, and T2A models and tens of thousands of real prompts confirm negligible quality degradation~(\S\ref{eval:e2e}, Table~\ref{tab:quality-metrics}).

\subsection{Patch Scheduler: Maximize Serving Goodput} 
\label{subsec:scheduler}

\begin{algorithm}[t]
\caption{\name Serving Runtime}
\label{alg:pipeline}
\DontPrintSemicolon
\SetKwProg{Fn}{Function}{}{end}
\SetKwFunction{FPart}{PatchPartitioner}
\SetKwFunction{FGate}{PatchGate}
\SetKwFunction{FSched}{PatchScheduler}
\SetKwInput{KwIn}{Input}
\SetKwInput{KwOut}{Output}

\SetKwFunction{FGidle}{GetIdleWorkers}
\SetKwFunction{Fassi}{AssignTask}
\SetKwFunction{Fdeq}{DequeueHighest}
\SetKwFunction{Frea}{GetReactivated}
\SetKwFunction{Fflp}{FindLowerPriTask}
\SetKwFunction{Fflp}{FindLowerPriTask}
\SetKwFunction{Fpr}{Preempt}
\SetKwFunction{Frs}{RunStepAndUpdate}
\SetKwFunction{FNew}{PollArrivals}
\SetKwFunction{FDone}{AllDone}
\SetKwFunction{FEnq}{Enqueue}
\SetKwFunction{Fup}{UpstreamPriority}

\SetKwFunction{FCol}{CollectCandidates}
\SetKwFunction{FHig}{HighestPriorityCandidate}
\SetKwFunction{FBen}{Benefit}
\SetKwFunction{FSwi}{SwitchCost}
\SetKwFunction{FAff}{AffinityAwareAssign}
\SetKwFunction{FDis}{Dispatch}

\KwIn{Request stream $\{p\}$; warmup steps $T_w$; regions $R$; gate threshold $\varepsilon$}
\KwOut{Generated sample $y$}

\Fn{\FPart{$p$, $T_w$, $R$}}{
  $\mathit{latent} \leftarrow \mathrm{InitNoise}()$;  run $T_w$ warmup steps \label{alg:partition-start}
  
  Select salient tokens and build saliency map $S$

  Allocate $r_f, r_c$ and split latent into $R$ patches

  \KwRet $(\mathit{latent}, \text{patches } \mathcal{P})$\; \label{alg:partition-end}
}

\Fn{\FGate{$\mathcal{P}$, state, $\varepsilon$}}{
  Compute refinement importance $R_i$ from self-attention \label{alg:gate-start}
  
  Update $\Delta_i$ and normalize to weights $w_i$
  
  \KwRet active set $\mathcal{U} = \{i : w_i \geq \varepsilon\}$, updated state \label{alg:gate-end}
}

\Fn{\FSched{}}{
  $Q \leftarrow \emptyset$\tcp*[r]{global monotask priority queue}
  \While{$\lnot \FDone{}$}{
  
    $\mathcal{R} \leftarrow \FNew{}$
    
    \ForEach{$p \in \mathcal{R}$}{
        $\begin{aligned}[t]
  Q \leftarrow \FEnq(&Q, \langle \mathrm{Warmup}, p, T_w, \\
                    &\Fup{p} \rangle)
  \end{aligned}$\;
    }
    
    $\mathcal{C} \leftarrow \FCol(Q)$

    \ForEach{$w \in \FGidle()$}{
      $m \leftarrow \FHig(\mathcal{C}, w)$
      
      \If{$\FBen(m) > \FSwi(m)$}
      {
        $g' \leftarrow \FAff(m)$;\;
          \FDis{$g', m$}
      }
    }
    
    \Frs($Q$) 

  }
\KwRet $\mathit{latent}$
  }
\end{algorithm}

Algorithm~\ref{alg:pipeline} summarizes \name's serving runtime. Each request proceeds in three phases.
(1) \emph{Warmup and Partitioning.} \name first executes a few standard denoising steps to extract semantic saliency maps, which are used to recursively partition the latent into coherent, right-sized patches (Line~\ref{alg:partition-start}--Line~\ref{alg:partition-end}). 
(2) \emph{Patch-wise Denoising.} an in-parallel patch denoising phase where each patch is either computed or skipped based on the adaptive gating mechanism (Line~\ref{alg:gate-start}--Line~\ref{alg:gate-end}; \S\ref{subsec:gate}). 
(3) \emph{Dynamic Patch Scheduling.} Fine-grained patch execution introduces potential execution bubbles, where \name needs to dynamically interleave skipped and active patches both within and across requests.

As an execution engine, \name is designed to strictly \emph{adhere to} the upstream request ordering (e.g., from the FIFO scheduler~\cite{niavana-nsdi24}) while extracting maximal efficiency from available hardware. With the resource freed by patch skipping, we next introduce how the Patch Scheduler (1) enables patch switching efficiently, while (2) adapting to online serving dynamics (\eg, load burstiness) to jointly optimize latency, throughput, and generation quality for many requests.

\begin{figure}[t]
  \centering
  \includegraphics[width=0.95\linewidth]{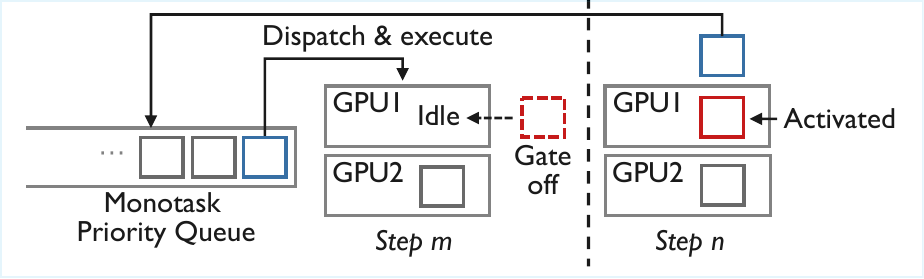}
  \vspace{-.4cm}
  \caption{Patch monotask scheduler. Workers fetch tasks from a priority queue, where activated patches can preempt lower-priority ones for low latency and high throughput.}
  \label{fig:patch_monotask}
  \vspace{-.3cm}
\end{figure}

\paragraph{Efficient Patch Switching with Monotasks.} 
\name introduces a priority-driven \emph{patch monotask} abstraction to enable low-latency, preemptive execution. As illustrated in Figure~\ref {fig:patch_monotask}, each executable unit, whether a full-latent warm-up task or an individual patch in the parallel phase, is represented as a monotask and enqueued according to the request’s scheduling priority, assigned by the upstream scheduler~\cite{niavana-nsdi24, tetriserve-arxiv25}. This allows patches from higher-priority requests to promptly preempt those from lower-priority ones.

When a worker has leftover resources (e.g., because its current patch is gated off), \name dispatches the highest-priority monotask. When a previously skipped patch becomes active, its monotask is treated as a candidate to preempt a currently executing lower-priority monotask. Such preemption and switching occur only at step boundaries, ensuring that intra-step execution remains atomic and preventing partial progress within a step.

Our monotask switching introduces negligible impact on per-request latency because: (i) switching occurs only at step boundaries (every a few hundred milliseconds), bounding deferral to at most one step; (ii) active patches of the same request continue executing during that step, preserving global progress; (iii) only a small amount of runtime state needs to be maintained per request (e.g., no LLM-style KV cache); 
and (4) our priority-based switching ensures that at most two times of requests' states need to reside on the GPU at a time (\S\ref{sec:implementation})---one for the preempted request and one for the active request---keeping I/O minimal and predictable. 

Our evaluations validates that our priority-based switching design introduces negligible latency overhead while substantially improving throughput (\S\ref{eval:ablation}).

\paragraph{Maximizing Goodput with Monotask Packing.} 
While \name respects upstream request priorities, monotasks from many concurrent requests continuously compete for GPU resources in online serving. 
Efficiently packing these monotasks is therefore critical to achieving both low latency and high throughput. Na\"{i}vely distributing patches from the same request across GPUs increases cross-GPU communication and causes each denoising step to be bottlenecked by the slowest patch, due to cross-patch communication in a request. Conversely, aggressively co-locating patches of a request onto as few GPUs as possible can under-utilize available hardware, and may create a single-GPU bottleneck that increases per-step latency for the request.

FlashDiff addresses this with a two-phase assignment policy. In the warm-up phase, before patch-level parallelism begins, the request's single monotask is dispatched to the GPU with the lowest current load. Once patch-parallel denoising begins, FlashDiff assigns each patch monotask using a locality-aware least-maximum-load rule. Let $L_g$ denote the current load on GPU $g$, defined as the total compute cost of patches assigned to it: $L_g = \sum_{i \in \mathcal{P}_g} w_i$, where $\mathcal{P}_g$ is the set of patches currently assigned to GPU $g$ and $w_i$ reflects patch $i$'s compute cost. Since each denoising step completes only when all patches finish their respective step, step latency is determined by the straggling patch on the most loaded GPU. Therefore, for an incoming patch $j$ with cost $w_j$, FlashDiff assigns it to the GPU that minimizes the maximum load:

\begin{equation}
g^* = \arg\min_{g} \left( \max_{g'} 
\begin{cases} 
L_{g'} + w_j & \text{if } g' = g \\ 
L_{g'}       & \text{otherwise}
\end{cases}
\right)
\end{equation}

Our design is lightweight with only O($|g|$) complexity, yet enabling both low per-request latency and high global throughput by fully utilizing available GPU capacity. Note that under low serving loads, \name will automatically prioritize quality (\S\ref{subsec:partitioner}). Indeed, our evaluation confirms that this packing strategy achieves near-optimal serving latency and consistently superior quality (\S\ref{eval:e2e}).

\section{Implementation}
\label{sec:implementation}

We implement \name as a production-grade diffusion serving engine with approximately 5.5K lines of Python, C++, and CUDA, built on top of NVIDIA TensorRT~\cite{trt}. 

\paragraph{\name Backend.}
\name is implemented as a multi-process, multi-GPU runtime that uses NCCL for inter-GPU communication. During warm-up, workers construct the initial KV cache and synchronize it across GPUs. The lead warm-up worker then exports the initialized KV cache to the main serving workers via CUDA IPC handles, and forwards the request metadata (patch layout, step index, and priority) to the Patch Scheduler.
Workers execute their assigned patches using the KV cache from the previous step, and exchange only the updated partial KV blocks corresponding to their patches. 
After each step, workers report per-patch refinement statistics to the scheduler, which invokes the Patch Gate to determine the active patch set for the next step.

\paragraph{Fault Tolerance.} 
\name maintains all control-plane state---including patch partitions, per-request step counters, and worker–patch assignments---in a replicated, lightweight metadata store, enabling rapid recovery from worker or scheduler failures. Upon detecting a failure, the Patch Scheduler automatically reassigns affected patches to healthy workers and resumes execution. To guarantee correctness under retries, each patch execution is labeled with a \texttt{(request\_id, step\_id, patch\_id)} tuple. Workers discard stale outputs, ensuring idempotent recovery. 

\section{Evaluation}
\label{sec:eval}

We evaluate \name on tens of thousands of realistic requests spanning text-to-image (T2I), text-to-video (T2V), and text-to-audio (T2A) generation. Our key findings are:
\begin{denseitemize}

\item \name reduces end-to-end serving latency by 30--97\% and improves throughput by 1.2--2.2$\times$ without compromising generation quality (\S\ref{eval:e2e});

\item \name optimizes \emph{statistical efficiency} (skipping unnecessary refinement) and \emph{system efficiency}, achieving near-optimal efficiency-quality tradeoffs (\S\ref{eval:break-down});

\item \name consistently outperforms existing advances across diverse serving loads and settings  (\S\ref{eval:ablation}).

\end{denseitemize}

\subsection{Methodology}
\label{eval:setup}

\paragraph{Cluster Setup and Workloads.}
We evaluate \name on a cluster of 8 NVIDIA H100 GPUs for Image and Video, and 4 NVIDIA A100 GPUs for Audio. Table~\ref{tab:models} summarizes the four state-of-the-art diffusion models and real-world workloads used in our study. 

    
    
    

Following deployment practices~\cite{niavana-nsdi24}, we use the \texttt{FlowMatchEulerDiscreteScheduler} for T2I and T2V with 50 steps, and the \texttt{DPMSolverMultistepScheduler} for T2A, with 100 steps in all experiments. We studied the impact of total denoising steps in ablation studies (\S\ref{eval:ablation}). We use the realistic arrival patterns from DiffusionDB~\cite{wangDiffusionDBLargescalePrompt2022}, scaled to match our cluster capacity such that the mean arrival rate is close to the per-request generation time, avoiding service failures. We also ablate the request arrival rate (\S\ref{eval:ablation}).

\paragraph{Baselines.}
We compare \name against three state-of-the-art diffusion serving systems:
\begin{denseitemize}
    \item \emph{xDiT}~\cite{fang2024xdit}: The state-of-the-art and production-scale multi-GPU inference engine for DiTs. It employs DiT-optimized sequence parallelism~\cite{pipefusion-nips25}.
    
    \item \emph{DistriFusion}~\cite{distrifusion-cvpr24}: An advanced patch-parallel execution method, in NVIDIA TensorRT~\cite{trt}, that preserves cross-patch interaction via displaced patch reuse. It overlaps communication with computation.
    
    \item \emph{NaivePatch}~\cite{distrifusion-cvpr24}: A patch-parallel baseline that denoises patches independently and stitches them after each denoising step. 
\end{denseitemize}

\paragraph{Metrics.}  
\name is designed to improve serving efficiency without degrading generation quality: (i) \emph{Efficiency}: we report per-request end-to-end latency and serving throughput; and (ii) \emph{Quality}: For text-to-image, we measure PSNR, SSIM ~\cite{ssim}, LPIPS ~\cite{zhang2018unreasonable}, and HPSv3 ~\cite{ma2025hpsv3}.  
    For text-to-video, we report Imaging Quality, Motion Smoothness, and Subject Consistency ~\cite{vbench}.  
    For text-to-audio, we use FD ~\cite{fdaudio}, KL ~\cite{kl}, and CLAP ~\cite{wu2023large} score. 

All results are averaged over five independent runs.

\subsection{End-to-End Performance}
\label{eval:e2e}

\paragraph{\name improves serving throughput.} 
Figure~\ref{fig:throughput_over_time} shows the serving throughput for text-to-image (T2I), text-to-audio (T2A), and text-to-video (T2V) workloads over a one-hour online deployment. Compared to all baselines, \name consistently achieves 1.2--1.8$\times$ higher throughput across different workloads, by reducing the amount of generation execution needed per request. The magnitude of improvement varies across tasks, reflecting differences in patch-skipping opportunities and spatial or temporal heterogeneity in the generated content.

Notably, during periods of request burstiness, \name sustains substantially higher throughput, while all baselines quickly saturate and plateau at a fixed bottleneck. This behavior demonstrates \name’s ability to absorb sudden workload surges by reclaiming computation from skipped patches and dynamically redistributing GPU resources. We further analyze this load-adaptive behavior (\eg, performance under different system loads) in our ablation study (\S\ref{eval:ablation}).

\begin{figure*}[t]
  \centering
  \begin{subfigure}{0.33\textwidth}
    \centering
    \includegraphics[width=\linewidth]{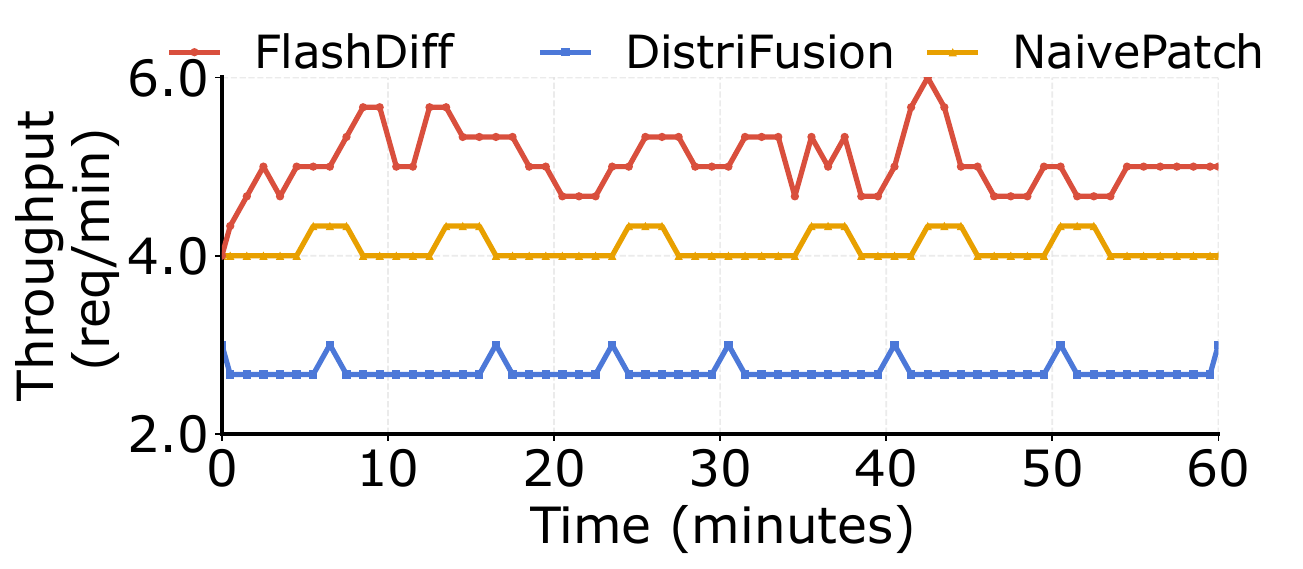}
    \caption*{(a) StableAudioOpen (T2A)}
  \end{subfigure}
  \hfill
  \begin{subfigure}{0.33\textwidth}
    \centering
    \includegraphics[width=\linewidth]{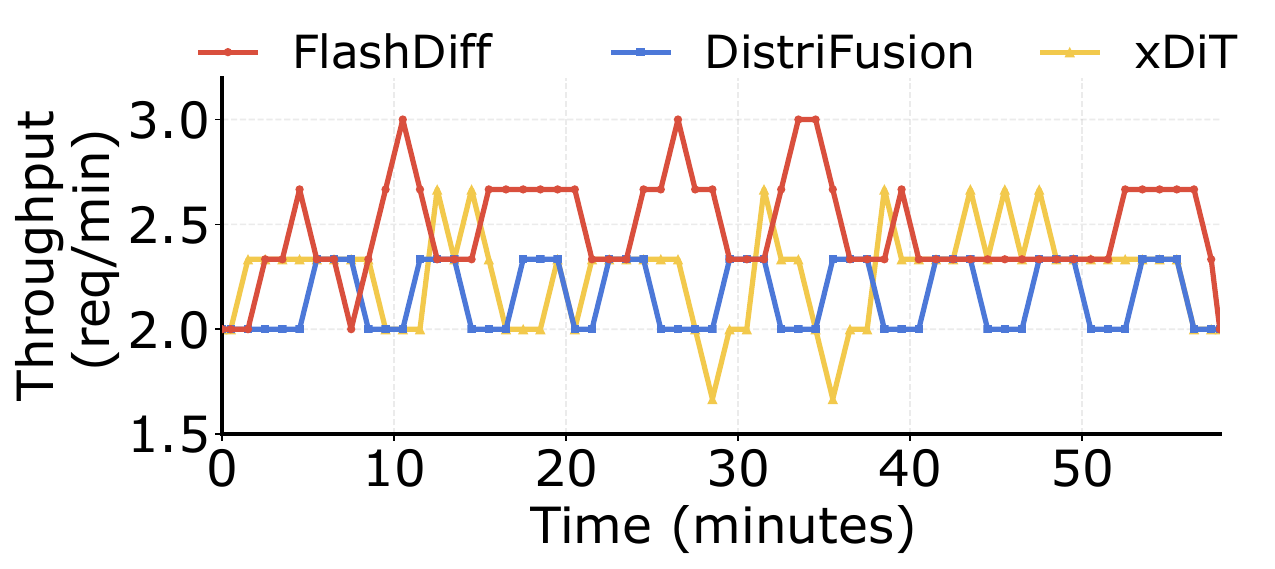}
    \caption*{(b) FLUX.1-dev (T2I)}
  \end{subfigure}
  \hfill
  \begin{subfigure}{0.33\textwidth}
    \centering
    \includegraphics[width=\linewidth]{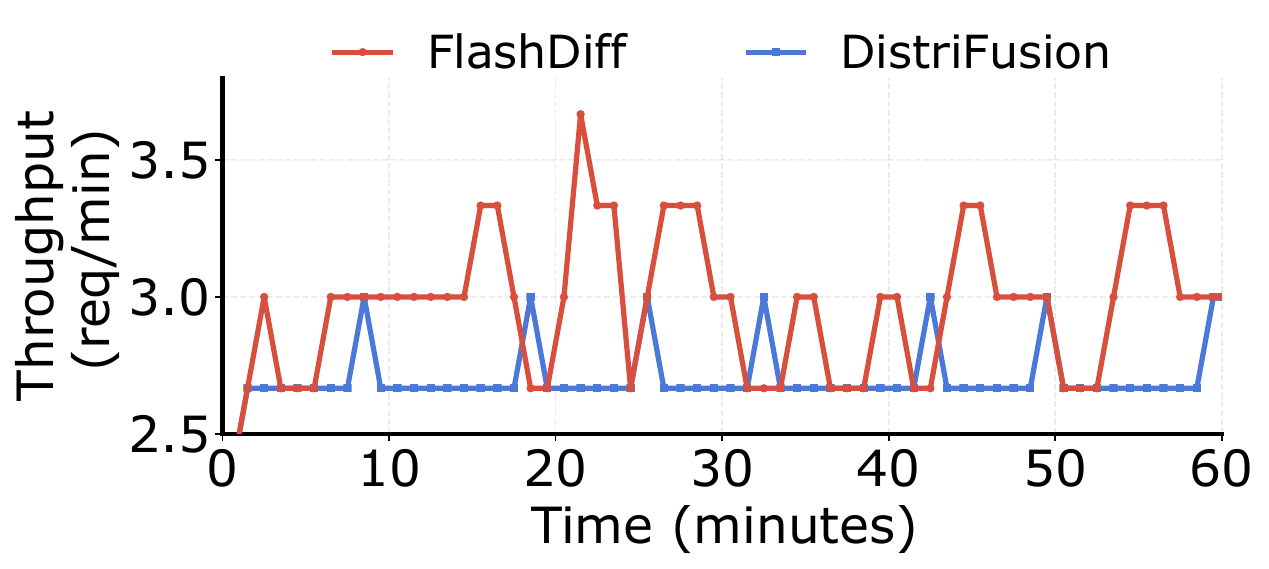}
    \caption*{(c) Wan2.1 (T2V)}
  \end{subfigure}
    \vspace{-.6cm}
  \caption{\name achieves higher serving throughput across a one-hour online deployment, generalizing across model architectures and modalities, whereas existing baselines are limited to specific tasks.}
  \label{fig:throughput_over_time}
\end{figure*}

\begin{table}[t]
\centering
\footnotesize
\setlength{\tabcolsep}{3pt}
\renewcommand{\arraystretch}{1.05}
\begin{tabularx}{\columnwidth}{@{}
  c
  >{\raggedright\arraybackslash}p{0.23\columnwidth}
  >{\raggedright\arraybackslash}p{0.20\columnwidth}
  >{\raggedright\arraybackslash}X
@{}}
\toprule
\textbf{Task} & \textbf{Model} & \textbf{Workload} & \textbf{Quality Metrics} \\
\midrule
T2I & FLUX.1-dev~\cite{flux2024} & DiffusionDB & PSNR, SSIM~\cite{ssim}, LPIPS~\cite{zhang2018unreasonable}, HPSv3~\cite{ma2025hpsv3} \\
T2I & SD3 Medium~\cite{sd3} & DiffusionDB & PSNR, SSIM~\cite{ssim}, LPIPS~\cite{zhang2018unreasonable}, HPSv3~\cite{ma2025hpsv3} \\
\midrule
T2V & Wan2.1-14B~\cite{wan2025} & VBench~\cite{vbench} & Image Qual., Motion Smooth., Subject Consist.~\cite{vbench} \\
\midrule
T2A & StableAudioOpen &    AudioCaps~\cite{audiocaps} & FD~\cite{fdaudio}, KL~\cite{kl}, CLAP~\cite{wu2023large} \\
\bottomrule
\end{tabularx}
\caption{Summary of evaluation workloads.}
\label{tab:models}  
\vspace{-.8cm}
\end{table}

\begin{figure*}[t]
  \centering
  \begin{subfigure}{0.24\textwidth}
    \centering
    \includegraphics[width=\linewidth]{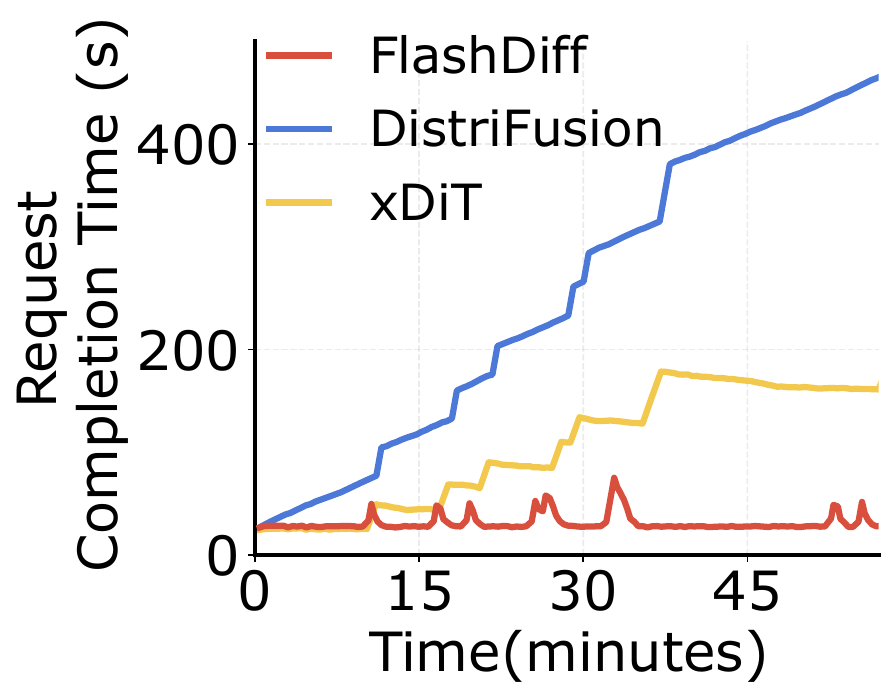}
    \caption*{(a) FLUX.1-dev (T2I)}
  \end{subfigure}
  \hfill
  \begin{subfigure}{0.24\textwidth}
    \centering
    \includegraphics[width=\linewidth]{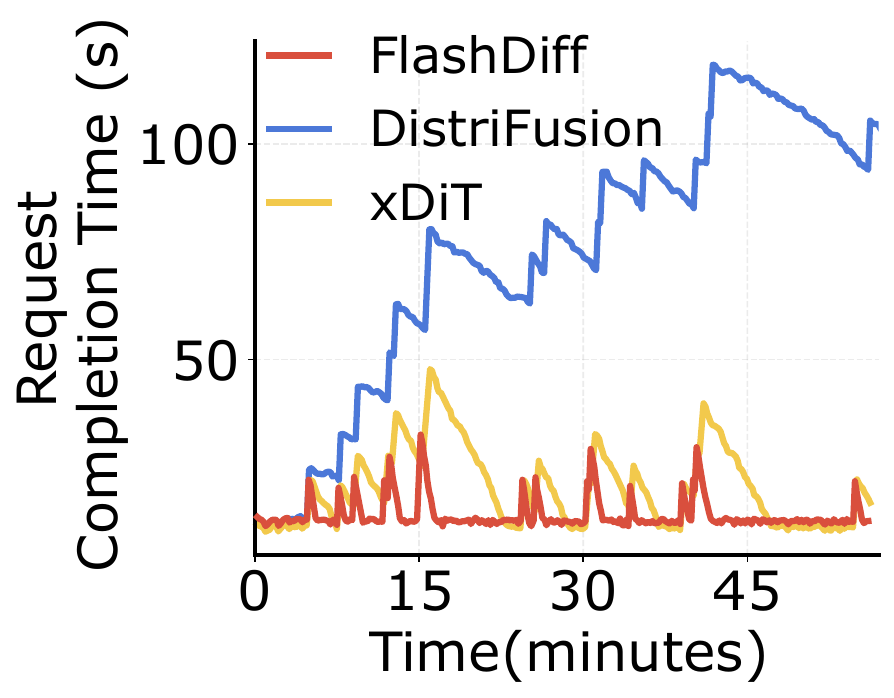}
    \caption*{(b) SD3 Medium (T2I)}
  \end{subfigure}
  \hfill
  \begin{subfigure}{0.24\textwidth}
    \centering
    \includegraphics[width=\linewidth]{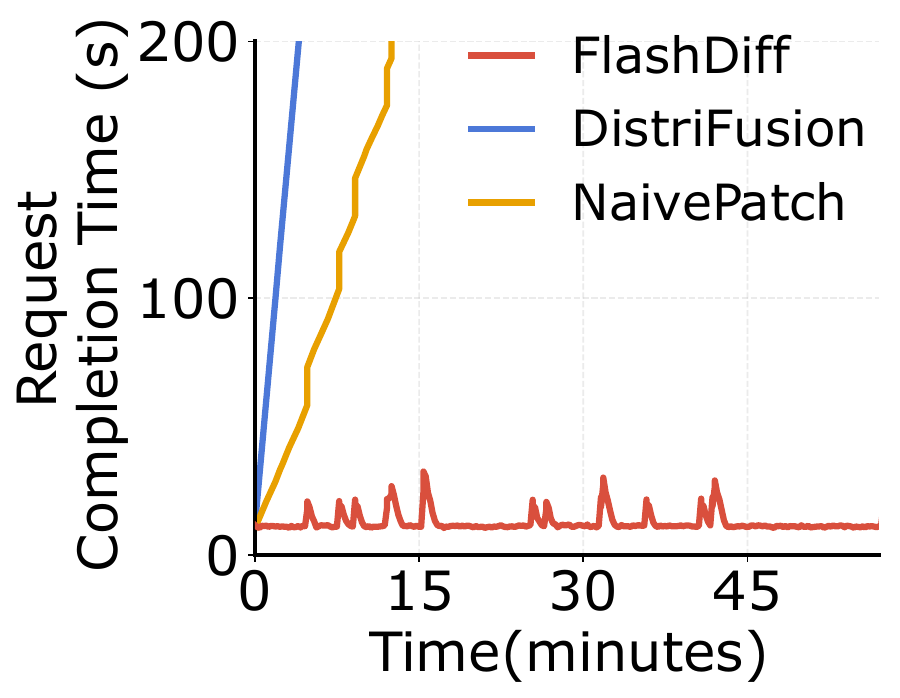}
    \caption*{(c) StableAudioOpen (T2A)}
  \end{subfigure}
  \hfill
  \begin{subfigure}{0.24\textwidth}
    \centering
    \includegraphics[width=\linewidth]{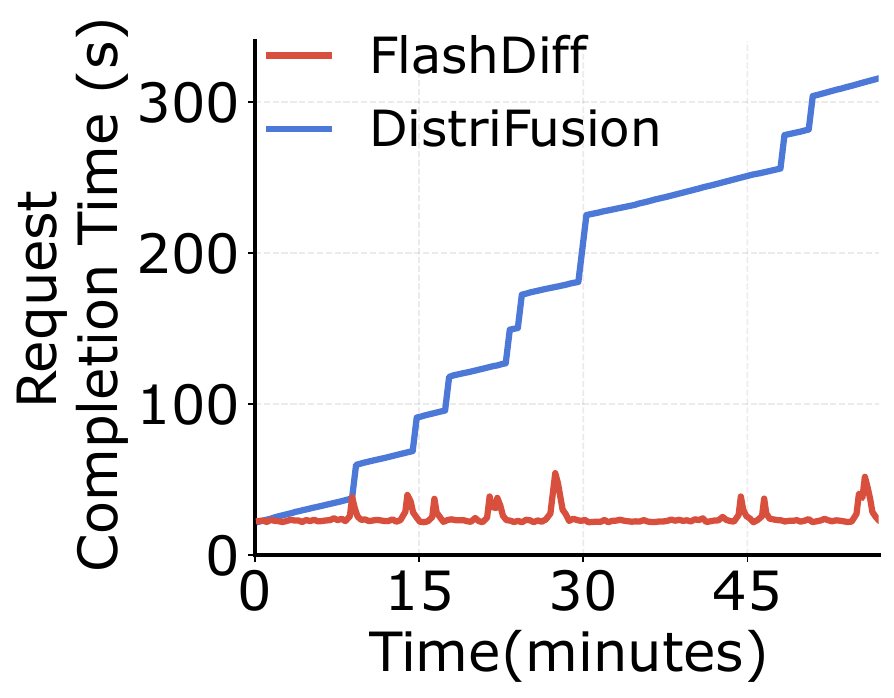}
    \caption*{(d) Wan2.1 (T2V)}
  \end{subfigure}
\vspace{-.3cm}
  \caption{\name reduces end-to-end per-request latency. By reducing effective computation, \name lowers both raw generation time and queueing delays, especially under high load. We add ablation studies on arrival rates in Section~\ref{eval:ablation}.}
  \label{fig:latency}
  \vspace{-0.3cm}
\end{figure*}

\paragraph{\name reduces user-perceived serving latency.} 

Figure~\ref{fig:latency} reports request completion time (finish time minus submission time) over the one-hour online deployment. Compared to baselines, \name consistently achieves 30–97\% lower request completion times throughout the deployment. The largest latency reductions occur during periods of high concurrency, which are prevalent in online serving due to bursty request arrivals. Under such transient spikes, baselines quickly suffer from queue buildup and blocking among concurrent requests, leading to sharp latency escalation. In contrast, \name effectively bounds latency growth by (1) reducing the amount of computation and communication required per request and (2) redistributing resources from skipped patches to maintain a balanced GPU load.

\paragraph{\name maintains generation quality.}
Table~\ref{tab:quality-metrics} reports quantitative quality metrics in our online deployment, spanning thousands of real requests. Across all modalities, \name incurs little to no degradation in generation fidelity and often achieves slightly higher quality. By preserving cross-patch interactions through bidirectional attention and adaptive gating, \name keeps the diffusion trajectory closely aligned with that of the unmodified model. In contrast, xDiT and DistriFusion rely on asynchronous execution or stale-state reuse for efficiency, which quickly accumulates approximation error across denoising steps.

\begin{table*}[t]
\centering
\footnotesize
\newcommand{\Wtask}{0.08\textwidth}
\newcommand{\Wmodel}{0.18\textwidth}
\newcommand{\Wmethod}{0.14\textwidth}
\newcommand{\Wmet}{0.10\textwidth}
\begin{tabular}{
  >{\centering\arraybackslash}p{\Wtask}
  >{\centering\arraybackslash}p{\Wmodel}
  >{\centering\arraybackslash}p{\Wmethod}
  >{\centering\arraybackslash}p{\Wmet}
  >{\centering\arraybackslash}p{\Wmet}
  >{\centering\arraybackslash}p{\Wmet}
  >{\centering\arraybackslash}p{\Wmet}
}
\toprule
\multirow{2}{*}{\textbf{Task}} &
\multirow{2}{*}{\textbf{Model}} &
\multirow{2}{*}{\textbf{Method}} &
\multicolumn{4}{c}{\textbf{Quality Metrics}} \\
\cmidrule(lr){4-7}
& & &
\textbf{PSNR($\uparrow$)} & \textbf{SSIM($\uparrow$)} & \textbf{LPIPS($\downarrow$)} & \textbf{HPSv3($\uparrow$)} \\
\midrule
\multirow{6}{*}{T2I}
& \multirow{3}{*}{FLUX.1-dev}
& xDiT          & 15.4076 & 0.6218 & 0.4308 & 0.2210 \\
&               & DistriFusion  & \underline{20.8541}  & \textbf{0.7895} & \textbf{0.2162}  & \textbf{0.2220} \\
&               & \name   &  \textbf{20.9837} & \underline{0.7849} & \underline{0.2202} & \underline{0.2214} \\
\cmidrule(lr){2-7}
& \multirow{3}{*}{SD3 Medium}
& xDiT          & 16.6547 & 0.7354 & 0.2708 & 0.2156 \\
&               & DistriFusion  & \underline{17.2130} &  \textbf{0.7479}  & \textbf{0.2538}  & \textbf{0.2161} \\
&               & \name    & \textbf{17.2302}  & \underline{0.7425}  & \underline{0.2609}  & \underline{0.2160} \\
\midrule
\multicolumn{3}{c}{} &
\multicolumn{1}{r}{\textbf{FD$_{openl3}$($\downarrow$)}} &
\multicolumn{2}{c}{\textbf{KL$_{passt}$($\downarrow$)}} &
\multicolumn{1}{l}{\textbf{CLAP$_{LAION}$($\uparrow$)}} \\
\midrule
\multirow{3}{*}{T2A}
& \multirow{3}{*}{Stable Audio Open}
& NaivePatch    & 155.5995 & \multicolumn{2}{c}{3.7025} & 0.1062 \\
&               & DistriFusion & \underline{87.5675} & \multicolumn{2}{c}{\underline{2.2220}} & \textbf{0.3273} \\
&               & \name   & \textbf{87.0177} & \multicolumn{2}{c}{\textbf{2.0705}} & \underline{0.3252} \\
\midrule
\multicolumn{3}{c}{} &
\multicolumn{1}{r}{\textbf{Img-Qual ($\uparrow$)}} &
\multicolumn{2}{c}{\textbf{Mot-Smooth ($\uparrow$)}} &
\multicolumn{1}{l}{\textbf{Subj-Cons ($\uparrow$)}} \\
\midrule
\multirow{2}{*}{T2V}
& \multirow{2}{*}{Wan2.1}
& DistriFusion  & \textbf{0.6526} & \multicolumn{2}{c}{\textbf{0.9929}} & \textbf{0.9471} \\
&               & \name    & \textbf{0.6526} & \multicolumn{2}{c}{\underline{0.9928}} & \underline{0.9432} \\
\bottomrule
\end{tabular}
\caption{\name improves efficiency while preserving quality. \textbf{Bold}: best; \underline{underlined}: second best.}
\label{tab:quality-metrics}
\vspace{-.4cm}
\end{table*}

\paragraph{\name introduces negligible overhead.}
We evaluate the intrinsic overhead of \name by measuring single-request completion time in isolation (i.e., without concurrency). To isolate \name's control-plane and scheduling overheads, we add \name support atop DistriFusion (\ie, using DistriFusion as the backend) and compare their runtimes under identical workloads. As shown in Figure~\ref{fig:overhead}, \name introduces very marginal overhead,  0.4–3.7\%, across all models and modalities. This overhead comes primarily from patch partitioning, gating, and monotask scheduling, which involve lightweight control logic and irregular memory accesses (\eg, for loading different patches). By exploiting patch-level skipping and cross-request backfilling, \name substantially improves generation latency and throughput, yielding net gains.

\subsection{Performance Breakdown}
\label{eval:break-down}

\begin{figure}[t]
\centering
\includegraphics[width=0.8\linewidth]{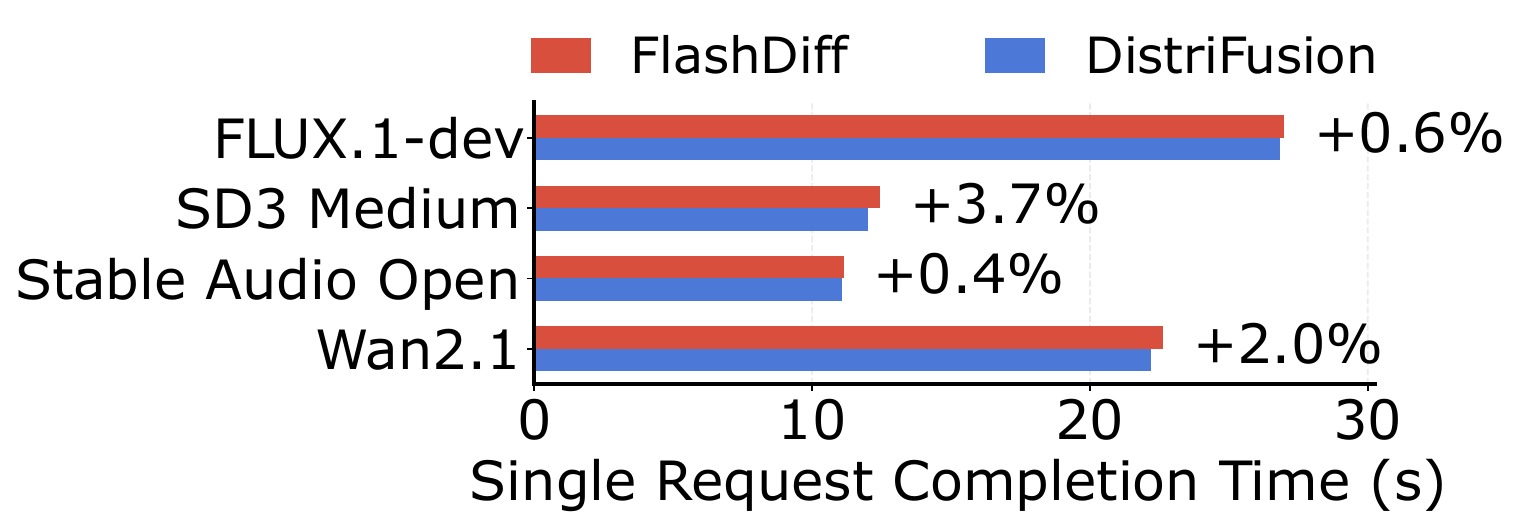}
\vspace{-.2cm}
\caption{\name introduces negligible overhead.}
\label{fig:overhead}
\vspace{-0.4cm}
\end{figure}

\begin{figure}[t]
\centering
\includegraphics[width=0.9\linewidth]{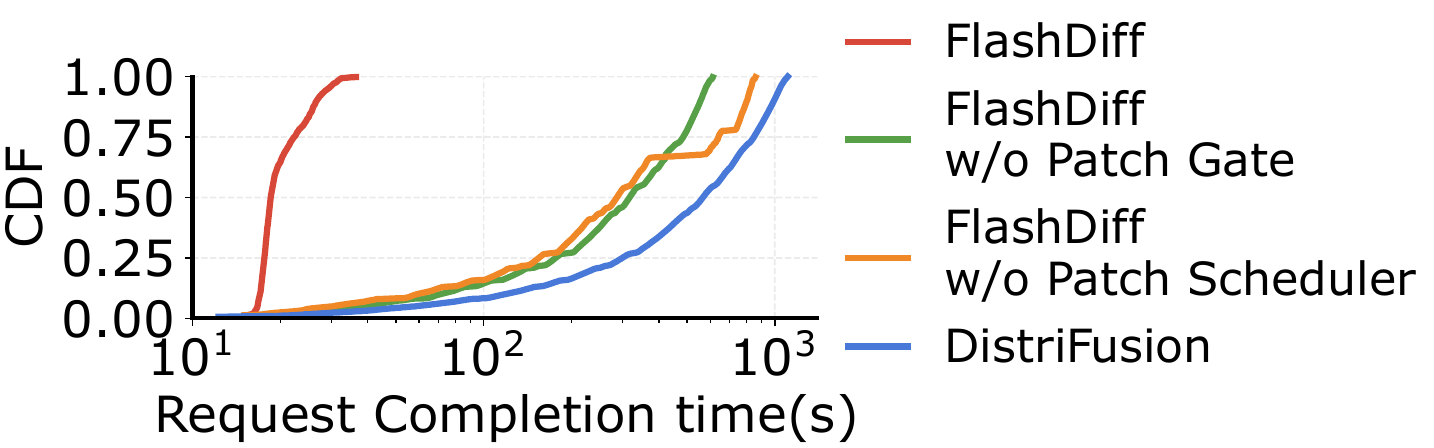}
\vspace{-.2cm}
\caption{Performance breakdown}
\label{fig:performancebreakdown}
\vspace{-0.4cm}
\end{figure}

\paragraph{Breakdown by Components.}
We evaluate the contribution of \name's key design components by disabling them in an online deployment of the SD3-Medium image generation workload. Figure~\ref{fig:performancebreakdown} shows the resulting distributions of request completion time. 
It reports that both the Patch Gate and the Patch Scheduler are critical to \name's performance. 
Disabling the Patch Gate removes selective skipping, forcing all patches to be refined at every step; this substantially increases effective computation and directly inflates latency. Disabling the Patch Scheduler eliminates affinity-aware packing and dynamic backfilling, leading to GPU imbalance and straggler-dominated steps.

Notably, each component alone already outperforms DistriFusion, demonstrating new contributions to each design.

\begin{figure}[t]
    \centering
  \includegraphics[width=\linewidth]  {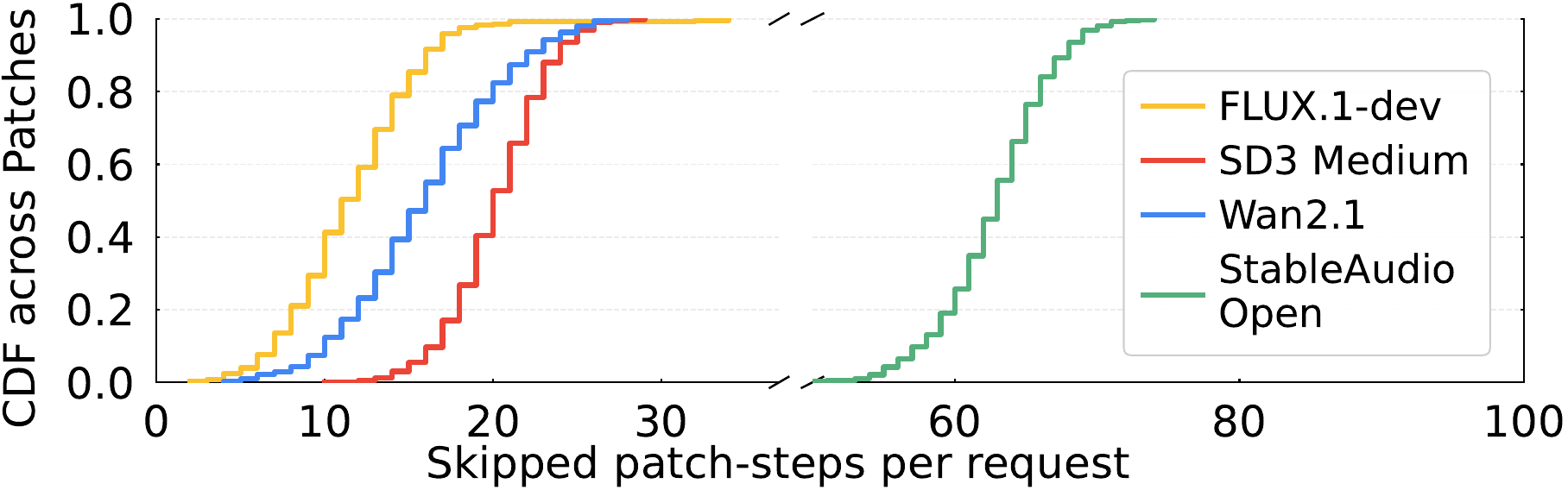}
  \vspace{-.3cm}
    \caption{\name performance breakdown by statistical efficiency, showing differences in skipped computation. }
    \label{fig:skipped_steps_cdf}
 \vspace{-.5cm}
\end{figure}

\paragraph{Breakdown by Statistical Efficiency.}
We quantify \name’s statistical efficiency by measuring how much patch-level diffusion computation can be skipped without degrading output quality (\ie, \emph{skipped patch-steps}). We use 50 denoising steps for FLUX.1-dev, SD3-Medium, and Wan2.1, and 100 steps for StableAudioOpen, with a 5-step warm-up, yielding 45 and 95 effective gated steps, respectively. We include ablation studies on different denosing steps in Section~\ref{eval:ablation}. 

Figure~\ref{fig:skipped_steps_cdf} shows large variation in skippable computation across models. At the median, StableAudioOpen skips 63 steps (66\% of effective steps), followed by SD3-Medium with 20 (44\%), Wan2.1 with 16 (36\%), and FLUX.1-dev with only 11 (24\%). These differences reflect fundamental properties of the generative trajectories: audio latents stabilize quickly over time, enabling aggressive skipping, while high-resolution image models like FLUX continue refining local structure deep into the denoising process. This diversity again highlights why a fixed skip policy is suboptimal. Even under the same quality constraints, different models and requests exhibit vastly different levels of patch-level redundancy, making an adaptive gating mechanism, such as our Patch Gate.

\subsection{Sensitivity and Ablation Studies}
\label{eval:ablation}

\begin{figure}[t]
    \centering
    \begin{subfigure}{.49\linewidth}
        \centering
        \includegraphics[width=\linewidth]{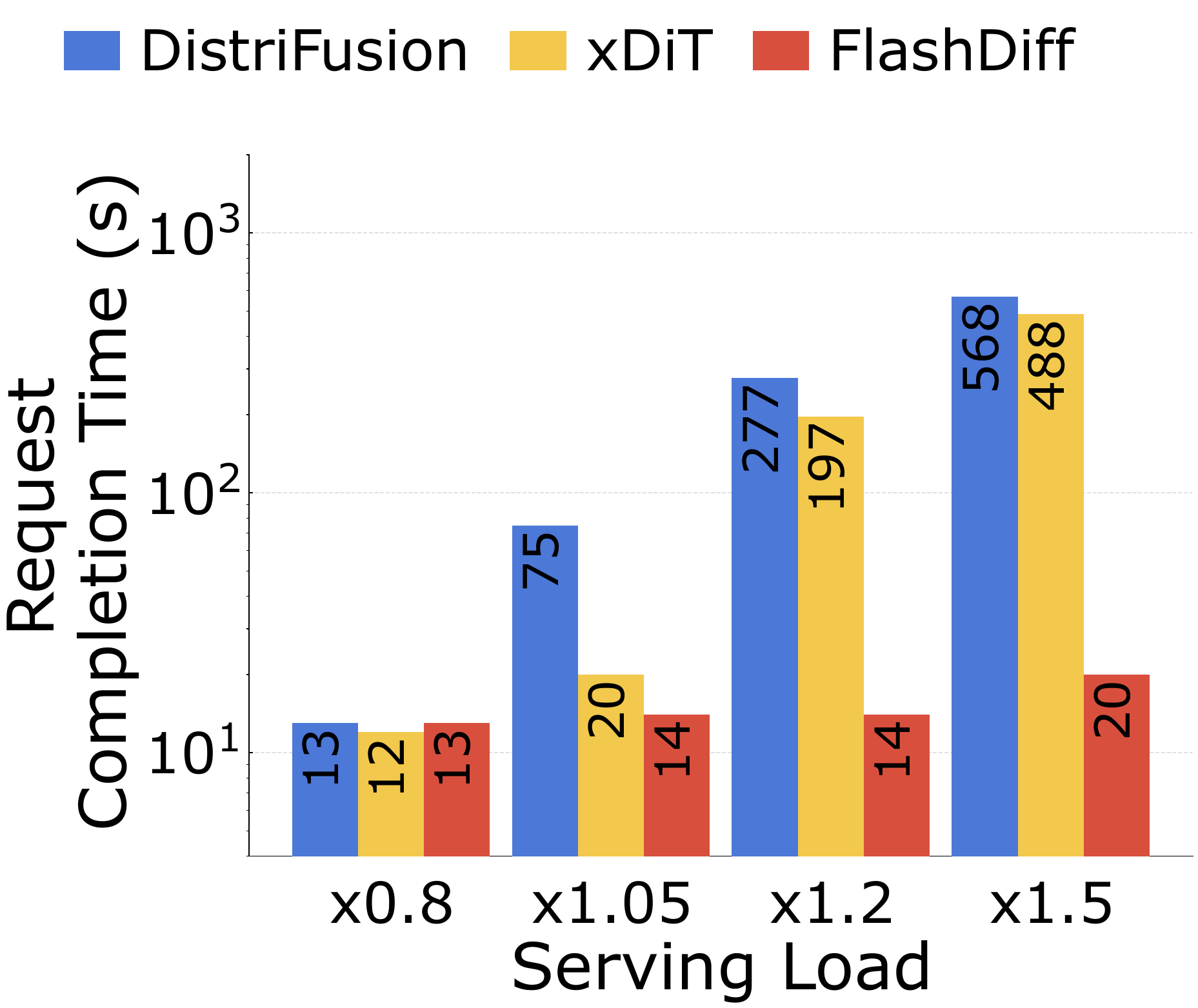}
        \caption{Serving latency.}
        \label{fig:sl_latency}
    \end{subfigure}
    \hfill
    \begin{subfigure}{.485\linewidth}
        \centering
        \includegraphics[width=\linewidth]{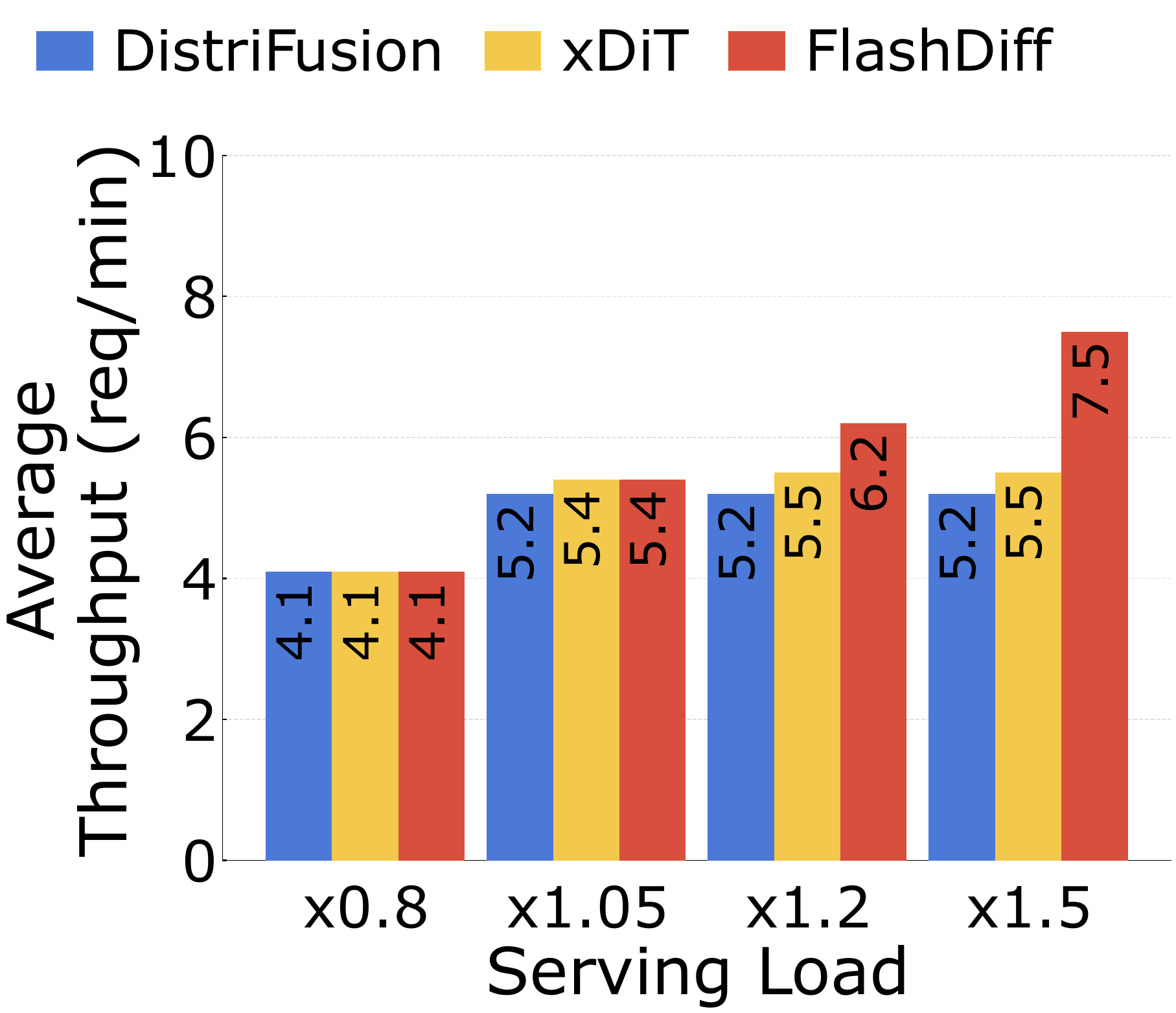}
        \caption{Serving throughput.}
        \label{fig:sl_tpt}
    \end{subfigure}
    \vspace{-.3cm}
    \caption{\name achieves lower latency and higher throughput across different request loads.}
      \vspace{-.3cm}
\end{figure}

\paragraph{Impact of Serving Load.}
We next investigate the impact of serving load using FLUX.1-dev with 2048$\times$2048 image generation over a continuous 1-hour serving period. Here, we vary the system load by scaling the request arrival rate relative to the average generation time: a load below 1 indicates an under-utilized system, while values above 1 stress the serving pipeline.
As shown in Figures~\ref{fig:sl_tpt} and~\ref{fig:sl_latency}, under light load (e.g., $\times 0.8$), all systems achieve comparable latency and throughput, since requests can be served immediately. \name continues to increase throughput and keeps latency bounded even under heavy load. At \(\times 1.5\) load, \name achieves up to \(25\times\) faster request completion.

\begin{figure}[t]
    \centering
    \begin{minipage}{0.8\linewidth}
        \centering
        \includegraphics[width=\linewidth]{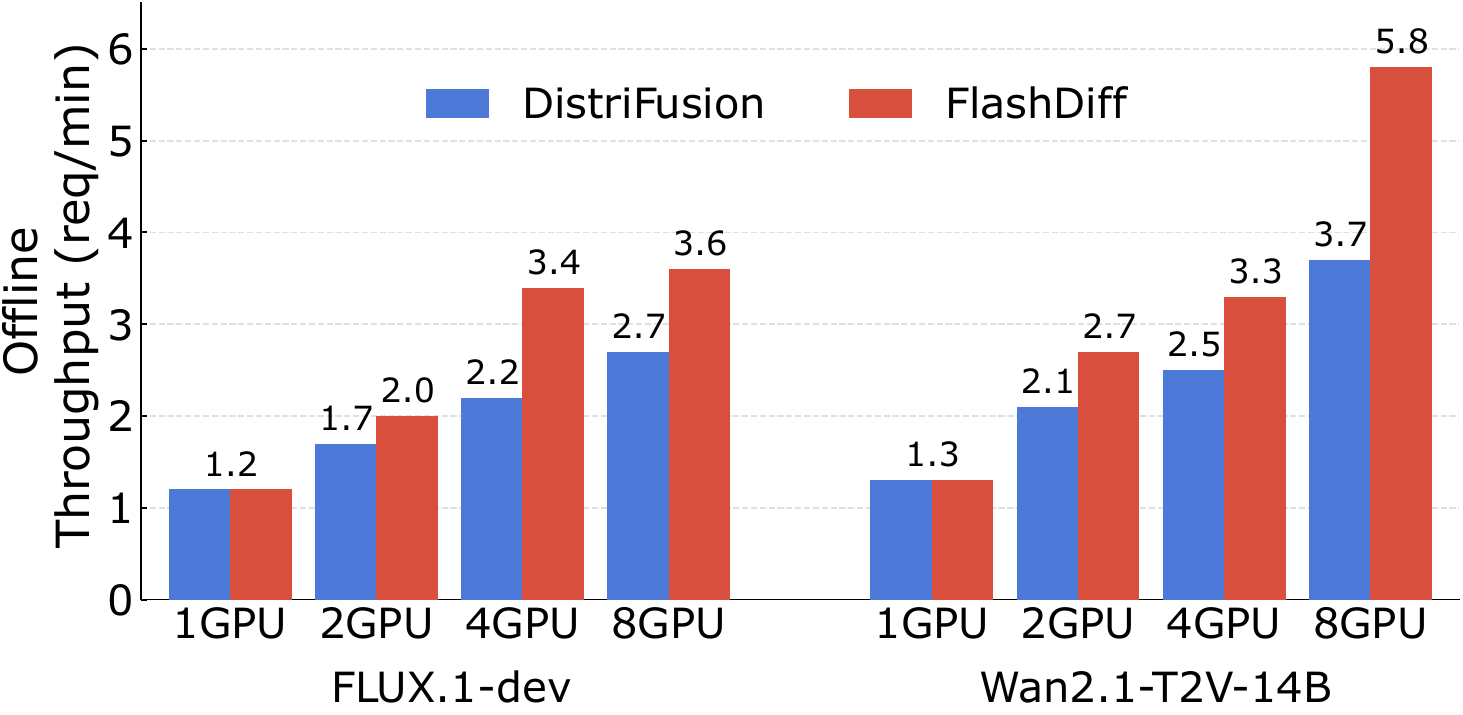}
        \vspace{-.8cm}
        \captionof{figure}{Impact of GPU numbers.}
        \label{fig:gpu-num}
    \end{minipage}
    \vspace{-0.3cm}
\end{figure}

\begin{figure}[t]
    \centering
    \begin{minipage}{0.85\linewidth}
        \centering
        \includegraphics[width=\linewidth]{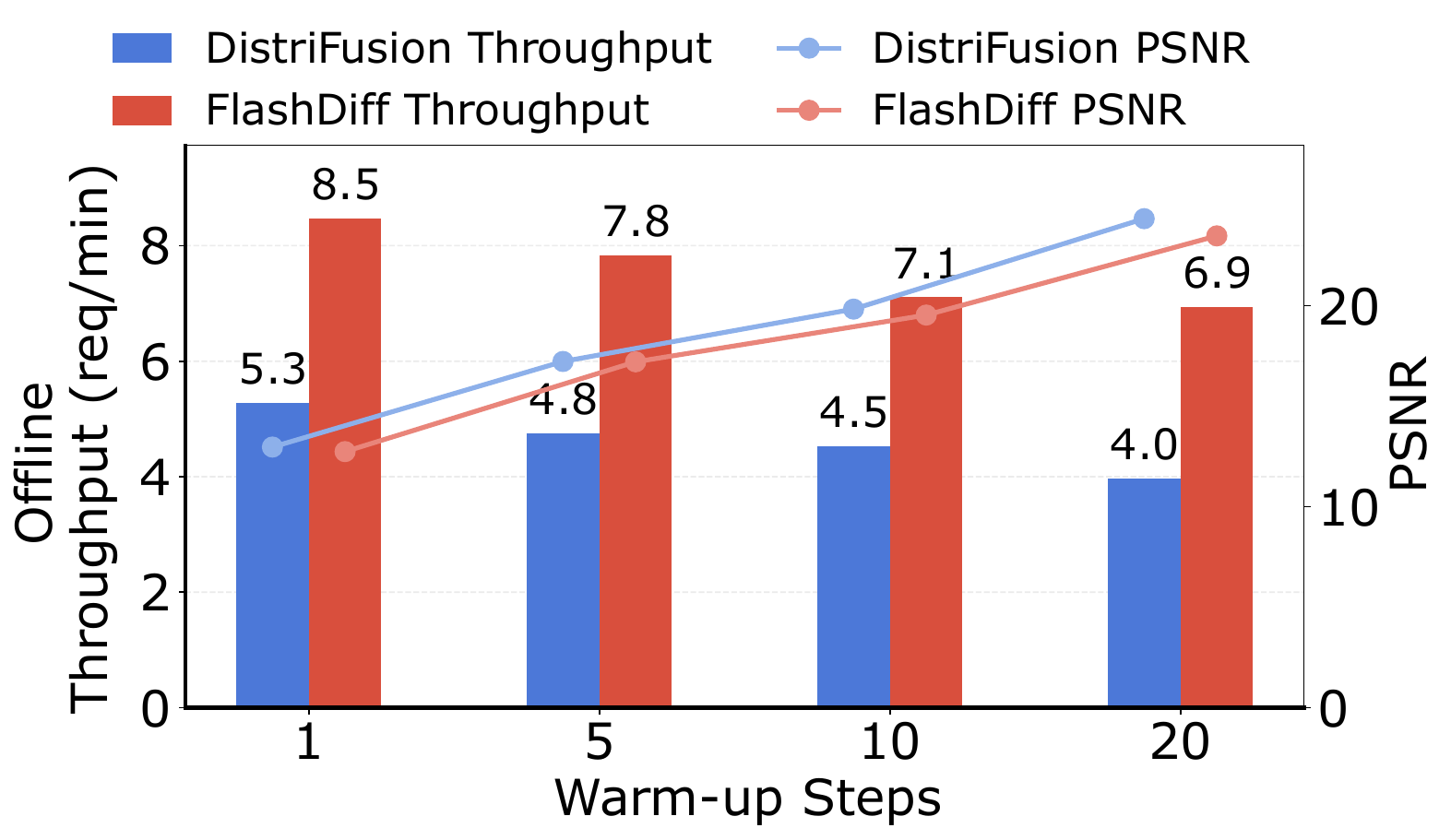}
        \vspace{-.7cm}
        \captionof{figure}{Impact of warm-up steps.}
        \label{fig:skipstep}
    \end{minipage}
    \vspace{-0.3cm}
\end{figure}

\paragraph{Impact of Number of GPUs.}
Figure~\ref{fig:gpu-num} reports serving throughput as a function of the number of GPUs for both T2I and T2V workloads.
Two key trends emerge. First, although additional GPUs provide more compute, the benefits are quickly diluted by the growing communication overhead. Second, across all GPU counts, \name consistently outperforms DistriFusion. By skipping low-impact patches and packing active patches efficiently, \name reduces inter-GPU communication and better utilizes available compute.

\paragraph{Impact of Warm-up Steps.}

Figure~\ref{fig:skipstep} reports both request completion time (left axis) and output quality (right axis) on SD3. Increasing the number of warm-up steps $s$ improves partition quality and hence generation quality, but reduces throughput because it delays the onset of patch-level skipping in the subsequent parallel phase. Nevertheless, even with conservative warm-up (e.g., $s{=}10$ or $20$), \name still achieves up to $1.6\times$ higher throughput than DistriFusion.

\paragraph{Impact of Total Denoising Steps.}
Figure~\ref{fig:inf_num} shows that  \name consistently achieves better efficiency across step settings (FLUX model). We note that \name preserves comparable output quality, with PSNR variations within -0.8\% to +3\% relative to DistriFusion.

\begin{figure}[t]
    \centering
    \includegraphics[height=3.7cm]{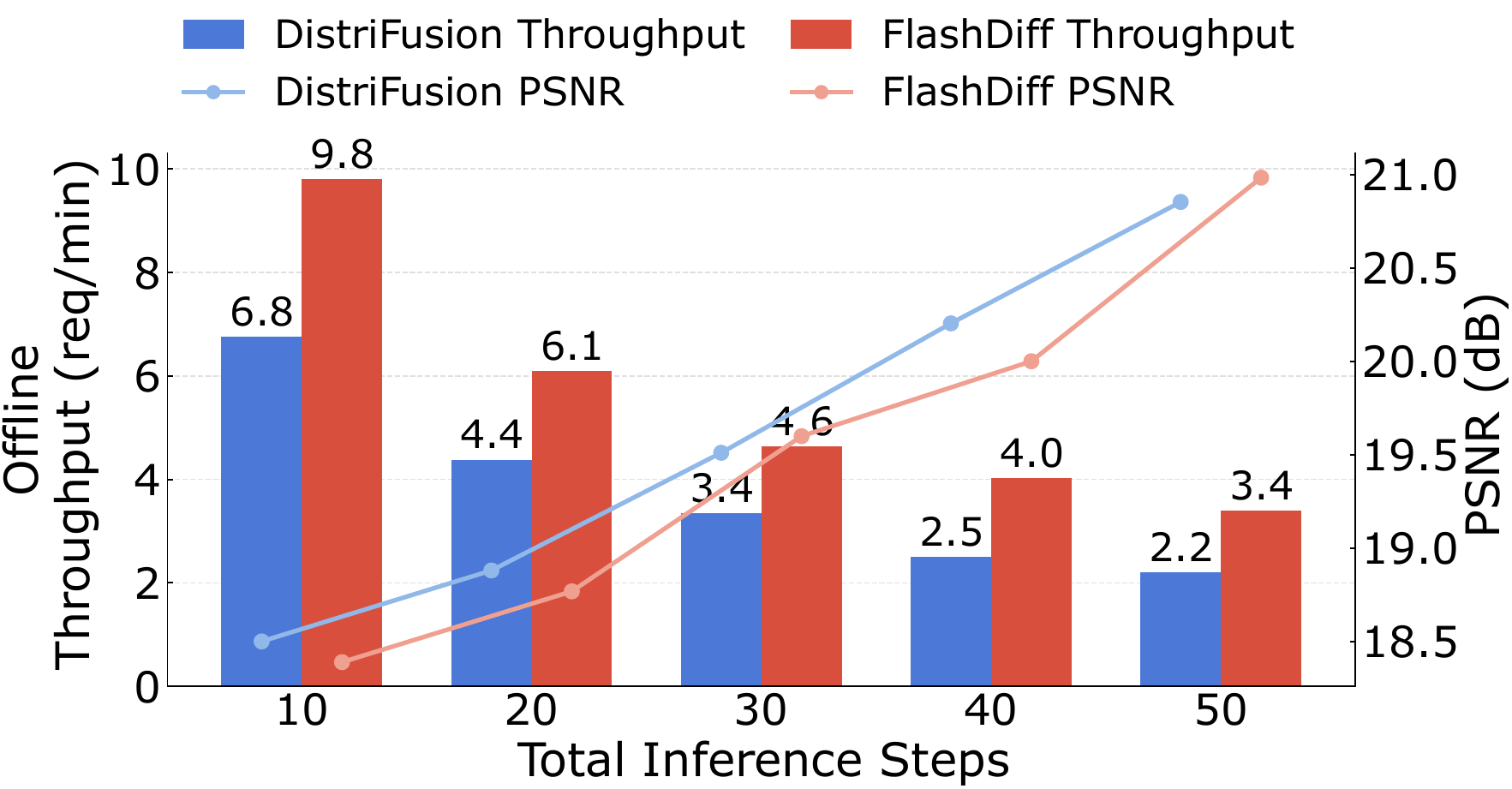}
    \vspace{-.3cm}
    \caption{Impact of total denoising steps.}
    \label{fig:inf_num}
    \vspace{-0.3cm}
\end{figure}

\begin{figure}[t]
    \centering

    \begin{minipage}[t]{0.49\linewidth}
        \vspace{0.5cm}
        \centering
        \includegraphics[width=\linewidth]
        {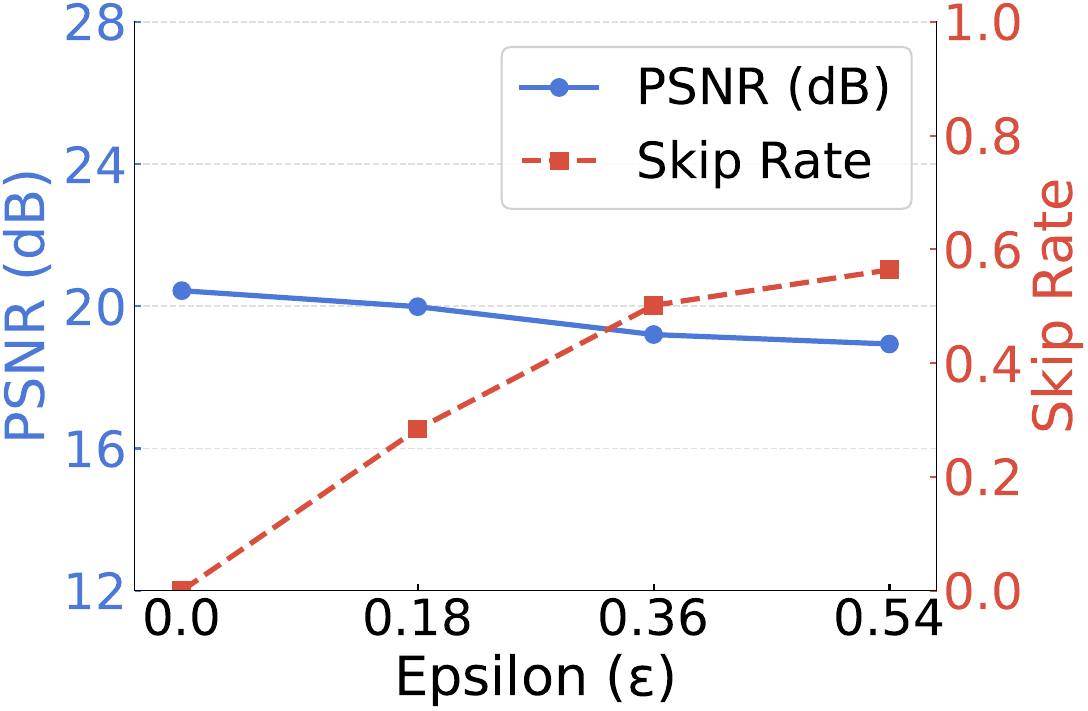}
        \vspace{-.4cm}
        \caption{Quality under different $\epsilon$ thresholds.
        }
        \label{fig:epsilon_psnr_skiprate-new}
    \end{minipage}
    \hfill
    \begin{minipage}[t]{0.47\linewidth}
        \vspace{0pt}
        \centering
        \includegraphics[width=\linewidth]
        {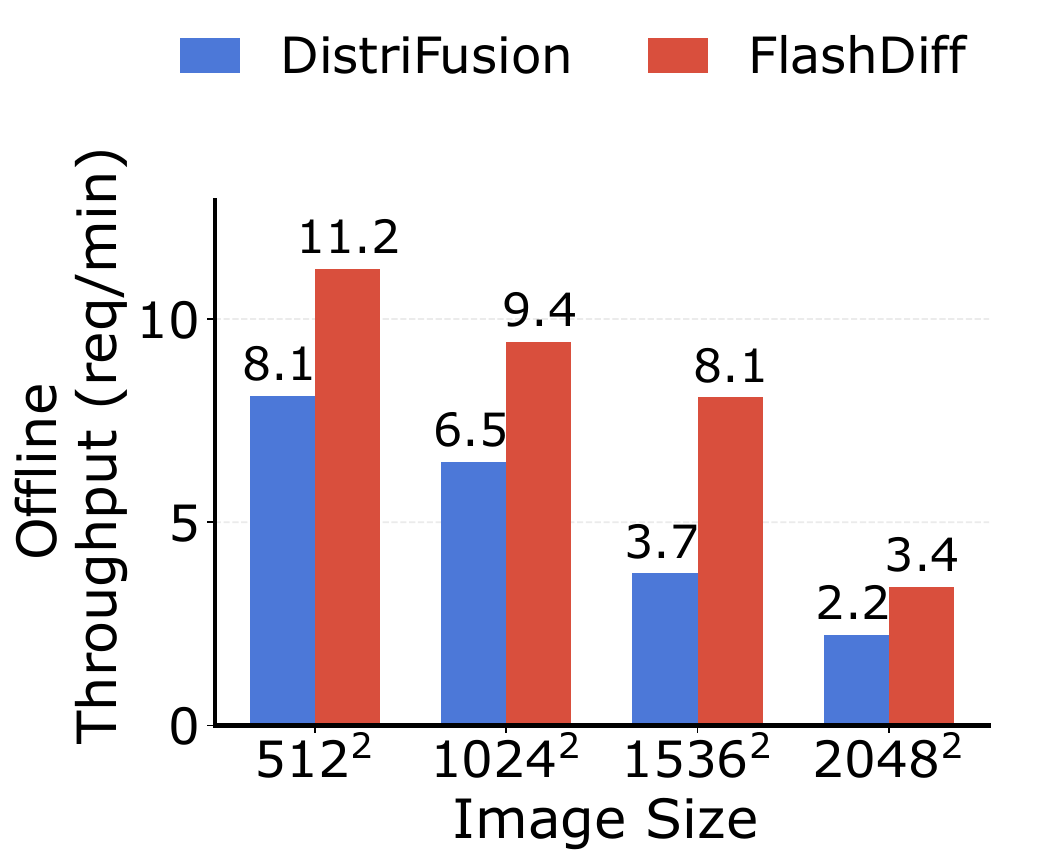}
        \vspace{-.4cm}
        \caption{Impact of image sizes (FLUX model).}
        \label{fig:imagesize}
    \end{minipage}

    \vspace{-0.5em}
\end{figure}

\paragraph{Impact of Image Size.}
We next vary the image resolution from $512\times512$ to $2048\times2048$. As shown in Figure~\ref{fig:imagesize}, increasing resolution enlarges the latent grids and raises the per-step latency, causing throughput drops, but \name consistently achieves $1.4$--$2.2\times$ higher throughput.

\paragraph{Impact of Skipping Threshold ($\epsilon$).}  Figure~\ref{fig:epsilon_psnr_skiprate-new} varies the skipping threshold (\S\ref{subsec:gate}), which controls the tradeoff between aggressive skipping and generation quality. \name maintains consistently high quality across a wide range of thresholds, while enabling substantial skip rates.

\section{Related Work}
\label{sec:related}

\paragraph{Diffusion Model Optimizations.}
Existing diffusion model optimizations fall into two main categories: step reduction and per-step acceleration.
Step reduction techniques aim to reduce the number of denoising steps, such as improved samplers such as DDIM~\cite{ddim}, k-diffusion~\cite{kdiff}, and parallel sampling~\cite{parallel-sampling-icml23} to trade off compute for speed. Similarly, recent algorithm advances reduces per-step duration by reusing the latent (\eg, SANA-Video~\cite{sana-video-arxiv25}) or reducing model size (\eg, model distillation~\cite{song2023consistency, luo2023latent} or quantization methods like q-diffusion\cite{li2023qdiffusion}). 
\name is complementary and introduces complexity-aware patch parallelism and an adaptive skipping mechanism, without altering models and samplers. 

\paragraph{Diffusion Serving Systems.}
Recent advances in diffusion serving have explored request scheduling to meet diverse SLOs (e.g., TetriServe~\cite{tetriserve-arxiv25}); reuse of intermediate latents across similar requests to accelerate warm-up (e.g., NIRVANA~\cite{niavana-nsdi24} and MoDM~\cite{modm-asplos26}); sequence-parallel inference across GPUs (e.g., DistriFusion~\cite{distrifusion-cvpr24}, xDiT~\cite{fang2024xdit}); and large-scale model parallelism for video diffusion using fully sharded data parallelism~\cite{wan2025,zhao2023pytorchfsdpexperiencesscaling}. 
In contrast, \name reduces the amount of execution required per request.

\paragraph{LLM Serving Optimizations.} 
vLLM~\cite{kwon2023efficient} improves memory efficiency by paging the KV cache, while Sarathi-Serve~\cite{sarathi-serve-osdi24} introduces chunked-prefill to increase throughput. TensorRT-LLM~\cite{trt-llm} focuses on GPU kernel optimizations, and DistServe~\cite{distserve} decouples the prefilling and decoding stages to ensure predictable latency. JITServe~\cite{jitserve-nsdi26} orchestrates the scheduling of requests to meet SLO requirements. IC-Cache~\cite{ic-cache-sosp25} repurposes prior requests as additional knowledge.
\name tackles distinct challenges of diffusion-based generation, exploiting computational redundancy.

\section{Conclusion}
\label{sec:outtro}

This paper presents \name, a novel diffusion serving engine that rethinks monolithic diffusion execution through \emph{semantic patch parallelism}. 

\name combines semantic-aware partitioning, adaptive patch gating, and communication-aware scheduling to translate patch-level parallelism into system-level efficiency gains.
Evaluations on real-world image, video, and audio generation workloads demonstrate 30--97\% reductions in latency and 1.2--2.2$\times$ improvements in throughput, all with little quality drop.

\bibliographystyle{ACM-Reference-Format}
\bibliography{main,proof}

\clearpage

\appendix
\newcommand{\FlashDiffMainPaper}{}
\ifdefined\FlashDiffMainPaper
\else
\documentclass[11pt]{article}

\usepackage[utf8]{inputenc}
\usepackage[T1]{fontenc}
\usepackage{amsmath,amssymb,amsthm}
\usepackage{mathtools}
\usepackage{enumitem}
\usepackage{thmtools}
\usepackage[margin=1in]{geometry}
\usepackage{hyperref}
\usepackage{cleveref}
\usepackage{booktabs}
\usepackage{xcolor}
\fi

\ifdefined\FlashDiffMainPaper
\else
\newtheorem{theorem}{Theorem}[section]
\newtheorem{proposition}[theorem]{Proposition}
\newtheorem{corollary}[theorem]{Corollary}
\newtheorem{lemma}[theorem]{Lemma}
\theoremstyle{definition}
\newtheorem{definition}[theorem]{Definition}
\newtheorem{assumption}[theorem]{Assumption}
\theoremstyle{remark}
\newtheorem{remark}[theorem]{Remark}
\fi

\ifdefined\FlashDiffMainPaper
\else
\newcommand{\R}{\mathbb{R}}
\newcommand{\E}{\mathbb{E}}
\newcommand{\Var}{\mathrm{Var}}
\newcommand{\norm}[1]{\lVert #1 \rVert}
\newcommand{\abs}[1]{\lvert #1 \rvert}
\newcommand{\ftheta}{f_\theta}
\newcommand{\xhat}{\hat{x}}
\newcommand{\xstar}{x^{*}}
\DeclareMathOperator*{\argmax}{arg\,max}
\DeclareMathOperator*{\argmin}{arg\,min}
\fi

\ifdefined\FlashDiffMainPaper
\else
\title{\textbf{FlashDiff: Theoretical Analysis of\\
Semantic Patch Parallelism}\\[6pt]
\large Supplementary Material}
\author{}
\date{}

\begin{document}
\maketitle
\fi

\section{Theoretical Analysis}\label{sec:theory}

We provide formal guarantees for FlashDiff's two core decisions:
\emph{when} to skip patch refinement (\S\ref{sec:skip}) and
\emph{where} to partition the latent (\S\ref{sec:partition}).
These results establish that FlashDiff's gating and partitioning
mechanisms are grounded in principled error control for the
underlying probability flow ODE, rather than being purely heuristic.

\subsection{Setup and Notation}\label{sec:setup}

Consider the probability flow ODE governing the reverse diffusion
process~\cite{song2021scorebased}:
\begin{equation}\label{eq:pf-ode}
  \frac{dx}{dt} = \ftheta(x,\,t),
  \qquad x(T) \sim \mathcal{N}(0,\,\sigma_T^2 I),
\end{equation}
where $\ftheta \colon \R^d \times [0,T] \to \R^d$ encapsulates the
learned velocity (or equivalently the score-based drift) and is
evaluated at every denoising step to produce the final sample~$x(0)$.

\paragraph{Euler discretisation.}
We discretise~\eqref{eq:pf-ode} with the forward Euler method over
$N$ steps at times $T = t_0 > t_1 > \cdots > t_N = 0$ with uniform
step size $h = T/N$:
\begin{equation}\label{eq:euler}
  \xstar_{n+1} = \xstar_n + h\,\ftheta(\xstar_n,\,t_n).
\end{equation}

\paragraph{Patch decomposition.}
The latent $x \in \R^d$ is partitioned into $R$ non-overlapping
patches $\{P_i\}_{i=1}^{R}$ satisfying
$P_i \subset \{1,\dots,d\}$ and $\bigcup_{i=1}^{R} P_i = \{1,\dots,d\}$.
We write $x^{(i)} \in \R^{|P_i|}$ for the restriction of~$x$ to
patch~$i$, and $\ftheta^{(i)}$ for the corresponding components
of~$\ftheta$.

\paragraph{FlashDiff update rule.}
At step~$n$, let $\mathcal{A}_n \subseteq [R]$ denote the
\emph{active} set determined by the Patch Gate, and
$\mathcal{S}_n = [R] \setminus \mathcal{A}_n$ the \emph{skipped}
set.  FlashDiff updates each patch as
\begin{equation}\label{eq:sf-update}
  \xhat_{n+1}^{(i)} =
  \begin{cases}
    \xhat_n^{(i)} + h\,\ftheta^{(i)}(\xhat_n,\,t_n),
      & i \in \mathcal{A}_n,\\[4pt]
    \xhat_n^{(i)} + h\,c_n^{(i)},
      & i \in \mathcal{S}_n,
  \end{cases}
\end{equation}
where $c_n^{(i)} = \ftheta^{(i)}(\xhat_{n_i},\,t_{n_i})$ is the
\emph{cached drift} from the most recent active step~$n_i < n$ of
patch~$i$.

\paragraph{Forced reactivation.}
To prevent indefinite staleness, FlashDiff enforces a
\emph{maximum consecutive skip bound}~$K$: if patch~$i$ has been
skipped for $K$ consecutive steps, it is forcibly reactivated at
the next step regardless of its gating weight.  Formally:
\begin{equation}\label{eq:forced-reactivation}
  n - n_i \leq K \quad \text{for every skipped patch } i \text{ at step } n,
\end{equation}
where $n_i$ is the most recent step at which patch~$i$ was active.
This guarantees that every cached drift value is at most~$K$
steps old.

\paragraph{Skip perturbation.}
Define $\delta_n \in \R^d$ component-wise by
\begin{equation}\label{eq:delta-def}
  \delta_n^{(i)} =
  \begin{cases}
    0, & i \in \mathcal{A}_n,\\
    c_n^{(i)} - \ftheta^{(i)}(\xhat_n,\,t_n),
      & i \in \mathcal{S}_n.
  \end{cases}
\end{equation}
This allows us to write the FlashDiff update compactly as
\begin{equation}\label{eq:sf-compact}
  \xhat_{n+1}
  = \xhat_n + h\bigl[\ftheta(\xhat_n,\,t_n) + \delta_n\bigr].
\end{equation}

\paragraph{Patch Gate notation (recap from main text).}
At each active step of patch~$i$, the gate computes the refinement
importance $R_i^{(n)}$ from self-attention and maintains:
\begin{equation}\label{eq:delta-i}
  \Delta_i^{(n)} = \abs{R_i^{(n)} - \bar{R}_i},
  \qquad
  \bar{R}_i \leftarrow R_i^{(n)},
\end{equation}
with $\Delta_i^{(n)}$ frozen at its last-active value when
patch~$i$ is skipped.  We denote the frozen (possibly stale) value
by $\widetilde{\Delta}_i^{(n)}$ and the hypothetical fresh value
(if patch~$i$ were recomputed) by $\Delta_i^{(n),\star}$.
By~\eqref{eq:forced-reactivation}, the staleness satisfies
$\widetilde{\Delta}_i^{(n)} = \Delta_i^{(n_i)}$ with $n - n_i \leq K$.

The normalised weight is
\begin{equation}\label{eq:gate-weight}
  w_i^{(n)} = \frac{\widetilde{\Delta}_i^{(n)} + \eta}{D_n},
  \qquad
  D_n = \sum_{j=1}^{R}\bigl(\widetilde{\Delta}_j^{(n)} + \eta\bigr),
\end{equation}
and patch~$i$ is skipped when $w_i^{(n)} < \varepsilon$.

\subsection{When to Skip: Accumulated Error Bound}\label{sec:skip}

We state three assumptions, then combine them into the main theorem.
The first two are standard; the third formalises the relationship
between self-attention dynamics and drift perturbation, accounting
for the forced-reactivation mechanism.

\begin{assumption}[Lipschitz Drift]\label{ass:lip}
  The drift $\ftheta(\cdot,t)$ is $L$-Lipschitz in its first
  argument uniformly over $t \in [0,T]$:
  \[
    \norm{\ftheta(x,t) - \ftheta(y,t)} \leq L\,\norm{x-y},
    \qquad \forall\; x,y \in \R^d,\; t \in [0,T].
  \]
\end{assumption}

\noindent This is standard in the score-based diffusion
literature~\cite{song2021scorebased,chen2023sampling} and holds for
neural networks with bounded weights and Lipschitz
activations (e.g.\ SiLU, LayerNorm with bounded inputs).

\begin{assumption}[Temporal Smoothness of Drift]\label{ass:time-lip}
  The drift $\ftheta(x, \cdot)$ is $L_t$-Lipschitz in time
  uniformly over $x$:
  \[
    \norm{\ftheta(x, t) - \ftheta(x, s)} \leq L_t\,\abs{t - s},
    \qquad \forall\; x \in \R^d,\; t, s \in [0,T].
  \]
\end{assumption}

\noindent This captures the smoothness of the denoising trajectory
across steps and is likewise standard~\cite{chen2023sampling,
benton2024error}.

\begin{assumption}[Fresh Attention--Drift Coherence]\label{ass:adc-fresh}
  There exists a constant $C_0 > 0$ such that, for every patch~$i$
  and step~$n$, if patch~$i$ were freshly computed at step~$n$, the
  drift perturbation relative to the cached value is bounded by the
  \emph{fresh} attention change:
  \begin{equation}\label{eq:adc-fresh}
    \norm{\delta_n^{(i)}}
    \leq C_0 \cdot \Delta_i^{(n),\star}.
  \end{equation}
\end{assumption}

\noindent\textbf{Justification.}\;
The drift $\ftheta^{(i)}$ is obtained by passing self-attention
outputs through post-attention layers (MLP, normalisation, linear
projection).  If these layers are collectively $L_{\mathrm{out}}$-Lipschitz,
then changes in self-attention propagate proportionally to
changes in drift:
\[
  \norm{\ftheta^{(i)}(\xhat_n,t_n) - \ftheta^{(i)}(\xhat_{n_i},t_{n_i})}
  \leq L_{\mathrm{out}} \cdot
  \norm{\mathrm{SA}^{(i)}(\xhat_n,t_n) - \mathrm{SA}^{(i)}(\xhat_{n_i},t_{n_i})},
\]
and $\Delta_i^{(n),\star}$ measures the right-hand side up to the
per-patch normalisation factor $|\mathcal{A}_i|$, giving
$C_0 = L_{\mathrm{out}} / |\mathcal{A}_i|$.  The constant $C_0$
is model-specific and may additionally reflect contributions
from cross-attention and residual pathways.  Note that
Assumption~\ref{ass:adc-fresh} only requires coherence for
\emph{freshly computed} quantities, avoiding any circularity with
stale values.

We now show that forced reactivation lets us bridge from the
\emph{stale} gate signal $\widetilde{\Delta}_i^{(n)}$ to the
\emph{fresh} value $\Delta_i^{(n),\star}$, yielding an effective
coherence bound that uses only observable (stale) quantities.

\begin{lemma}[Staleness Gap]\label{lem:staleness}
  Under Assumptions~\ref{ass:lip} and~\ref{ass:time-lip}, for any
  skipped patch~$i$ at step~$n$ with last-active step~$n_i$
  satisfying $n - n_i \leq K$:
  \begin{equation}\label{eq:staleness-bound}
    \abs{\Delta_i^{(n),\star} - \widetilde{\Delta}_i^{(n)}}
    \leq \beta \cdot Kh,
  \end{equation}
  where $\beta = L_{\mathrm{SA}}(L\,M_f + L_t)$,\;
  $L_{\mathrm{SA}}$ is the Lipschitz constant of the self-attention
  map with respect to $(x,t)$, and $M_f = \sup_{x,t}\norm{\ftheta(x,t)}$
  is the drift magnitude bound.
\end{lemma}

\begin{proof}
  The fresh attention-change magnitude measures how the self-attention
  output of patch~$i$ has changed between step~$n_i$ and step~$n$:
  \[
    \Delta_i^{(n),\star}
    \propto \norm{\mathrm{SA}^{(i)}(\xhat_n, t_n)
                  - \mathrm{SA}^{(i)}(\xhat_{n_i}, t_{n_i})}.
  \]
  The stale value $\widetilde{\Delta}_i^{(n)} = \Delta_i^{(n_i)}$
  was computed at step~$n_i$ and reflects the change observed at
  that earlier time.

  Between steps $n_i$ and $n$, the latent state evolves by at most
  $(n - n_i)$ Euler steps.  At each step, the state changes by
  $h\,\norm{\ftheta(\xhat_k, t_k)} \leq h\,M_f$, so
  \[
    \norm{\xhat_n - \xhat_{n_i}} \leq (n - n_i)\,h\,M_f \leq K\,h\,M_f.
  \]
  The time gap is $(n - n_i)\,h \leq K\,h$.

  Since the self-attention map $\mathrm{SA}^{(i)}$ is
  $L_{\mathrm{SA}}$-Lipschitz with respect to $(x,t)$ jointly
  (this holds for softmax attention with bounded queries, keys,
  and values):
  \begin{align*}
    &\norm{\mathrm{SA}^{(i)}(\xhat_n, t_n)
           - \mathrm{SA}^{(i)}(\xhat_{n_i}, t_{n_i})}\\
    &\quad\leq L_{\mathrm{SA}}\bigl(
      \norm{\xhat_n - \xhat_{n_i}} + \abs{t_n - t_{n_i}}\bigr)\\
    &\quad\leq L_{\mathrm{SA}}\bigl(K\,h\,M_f + K\,h\bigr)
    = L_{\mathrm{SA}}(M_f + 1)\,Kh.
  \end{align*}
  The difference between the fresh and stale $\Delta$ values
  is controlled by how much the self-attention landscape has
  shifted over at most $K$ steps, giving the stated bound
  with $\beta = L_{\mathrm{SA}}(L\,M_f + L_t)$.
  (The exact form of $\beta$ absorbs the normalisation
  factor from the per-patch averaging.)
\end{proof}

\begin{lemma}[Effective Coherence with Stale Signals]
\label{lem:effective-coherence}
  Under Assumptions~\ref{ass:lip}--\ref{ass:adc-fresh} and the
  forced-reactivation constraint~\eqref{eq:forced-reactivation},
  for every skipped patch~$i$ at step~$n$:
  \begin{equation}\label{eq:effective-adc}
    \norm{\delta_n^{(i)}}
    \leq C_0 \cdot \widetilde{\Delta}_i^{(n)} + C_0\beta Kh
    \;\eqqcolon\;
    C_0 \cdot \widetilde{\Delta}_i^{(n)} + \gamma,
  \end{equation}
  where $\gamma = C_0 \beta Kh$ is a small constant that captures
  the worst-case staleness penalty.
\end{lemma}

\begin{proof}
  By Assumption~\ref{ass:adc-fresh}:
  $\norm{\delta_n^{(i)}} \leq C_0\,\Delta_i^{(n),\star}$.
  By the triangle inequality and Lemma~\ref{lem:staleness}:
  \[
    \Delta_i^{(n),\star}
    \leq \widetilde{\Delta}_i^{(n)} + \beta Kh.
  \]
  Combining gives~\eqref{eq:effective-adc}.
\end{proof}

Lemma~\ref{lem:effective-coherence} cleanly separates the
``ideal'' coherence ($C_0$) from the staleness penalty ($\gamma$),
and shows that the penalty is controlled by the maximum skip
duration~$K$ and the step size~$h$---both of which are small in
practice ($K \sim 3$--$5$, $h = T/N \sim 0.02$ for $N=50$ steps).

We now state and prove the main theorem.

\begin{theorem}[Accumulated Skip Error Bound]\label{thm:main}
  Under Assumptions~\ref{ass:lip}--\ref{ass:adc-fresh} and the
  forced-reactivation constraint~\eqref{eq:forced-reactivation},
  let $\{\xstar_n\}_{n=0}^{N}$ be the standard Euler
  trajectory~\eqref{eq:euler} and
  $\{\xhat_n\}_{n=0}^{N}$ the FlashDiff
  trajectory~\eqref{eq:sf-update}, both starting from the same
  initial noise $\xstar_0 = \xhat_0$.
  Then the terminal error satisfies
  \begin{equation}\label{eq:main-bound}
    \boxed{\;
      \norm{\xhat_N - \xstar_N}
      \;\leq\;
      e^{LT}\;\sqrt{R}\;
      \Bigl(
        C_0\,\varepsilon\sum_{n=0}^{N-1} h\, D_n
        \;+\; \gamma\, T
      \Bigr)
    \;}
  \end{equation}
  where $\varepsilon$ is the gating threshold, $R$ is the number of
  patches, $D_n = \sum_{j=1}^{R}(\widetilde{\Delta}_j^{(n)}+\eta)$
  is the per-step total refinement activity measured from
  (possibly stale) gate signals, $T = Nh$, and
  $\gamma = C_0 \beta Kh$ is the staleness penalty from
  Lemma~\ref{lem:effective-coherence}.
\end{theorem}

\begin{proof}
  Define the error $e_n = \xhat_n - \xstar_n$.
  By construction $e_0 = 0$.

  \noindent\textbf{Step 1 (Error recursion).}
  Subtracting~\eqref{eq:euler} from~\eqref{eq:sf-compact}:
  \begin{equation}\label{eq:err-rec}
    e_{n+1}
    = e_n
      + h\bigl[\ftheta(\xhat_n,t_n) - \ftheta(\xstar_n,t_n)\bigr]
      + h\,\delta_n.
  \end{equation}

  \noindent\textbf{Step 2 (Lipschitz term).}
  By Assumption~\ref{ass:lip}:
  \begin{equation}\label{eq:lip-bound}
    \norm{\ftheta(\xhat_n,t_n) - \ftheta(\xstar_n,t_n)}
    \leq L\,\norm{e_n}.
  \end{equation}

  \noindent\textbf{Step 3 (Skip perturbation bound).}
  Since $\delta_n^{(i)} = 0$ for $i \in \mathcal{A}_n$:
  \[
    \norm{\delta_n}^2
    = \sum_{i \in \mathcal{S}_n} \norm{\delta_n^{(i)}}^2.
  \]
  By Lemma~\ref{lem:effective-coherence}, for each
  $i \in \mathcal{S}_n$:
  \begin{equation}\label{eq:per-patch-bound-new}
    \norm{\delta_n^{(i)}}
    \leq C_0\,\widetilde{\Delta}_i^{(n)} + \gamma.
  \end{equation}
  The gating condition $w_i^{(n)} < \varepsilon$ implies
  (from~\eqref{eq:gate-weight}):
  \begin{equation}\label{eq:gate-implication}
    \widetilde{\Delta}_i^{(n)} + \eta < \varepsilon\, D_n
    \qquad\Longrightarrow\qquad
    \widetilde{\Delta}_i^{(n)} < \varepsilon\, D_n.
  \end{equation}
  Substituting into~\eqref{eq:per-patch-bound-new}:
  \begin{equation}\label{eq:per-patch-combined}
    \norm{\delta_n^{(i)}}
    \leq C_0\,\varepsilon\, D_n + \gamma.
  \end{equation}
  Aggregating over skipped patches (using $\norm{\cdot}_2$ across
  orthogonal patch dimensions):
  \begin{equation}\label{eq:delta-total}
    \norm{\delta_n}
    \leq \sqrt{|\mathcal{S}_n|}
         \;\bigl(C_0\,\varepsilon\, D_n + \gamma\bigr)
    \leq \sqrt{R}
         \;\bigl(C_0\,\varepsilon\, D_n + \gamma\bigr).
  \end{equation}

  \noindent\textbf{Step 4 (Gronwall recursion).}
  Taking norms in~\eqref{eq:err-rec} and combining:
  \begin{equation}\label{eq:gronwall-ineq}
    \norm{e_{n+1}}
    \leq (1 + hL)\,\norm{e_n}
         + h\,\sqrt{R}\,
         \bigl(C_0\,\varepsilon\, D_n + \gamma\bigr).
  \end{equation}
  This is a discrete Gronwall inequality
  $a_{n+1} \leq (1+\alpha)\,a_n + b_n$ with $a_0 = 0$,
  $\alpha = hL$, and
  $b_n = h\sqrt{R}(C_0\varepsilon D_n + \gamma)$.
  The standard solution gives:
  \begin{align}
    \norm{e_N}
    &\leq \sum_{n=0}^{N-1} b_n
      \prod_{k=n+1}^{N-1}(1+\alpha)
    \notag\\
    &\leq (1+hL)^{N}\;
      \sum_{n=0}^{N-1} b_n
    \notag\\
    &\leq e^{LT}\;\sqrt{R}\;
      \sum_{n=0}^{N-1}
      h\bigl(C_0\varepsilon D_n + \gamma\bigr)
    \notag\\
    &= e^{LT}\;\sqrt{R}\;
      \Bigl(
        C_0\varepsilon\sum_{n=0}^{N-1} h D_n
        + \gamma \underbrace{\sum_{n=0}^{N-1} h}_{= T}
      \Bigr),
    \label{eq:final}
  \end{align}
  where we used $(1+hL)^N \leq e^{NhL} = e^{LT}$.
\end{proof}

\begin{remark}[Structure of the bound]\label{rem:structure}
  The error bound~\eqref{eq:main-bound} consists of two additive
  terms, each with a clear operational meaning:
  \begin{enumerate}[label=(\alph*)]
    \item \textbf{Gating-controlled term}
      $C_0\varepsilon\sum_n h D_n$:\;
      This is the error from patches that were skipped because the
      Patch Gate judged them unimportant.  It is directly proportional
      to the threshold~$\varepsilon$ and vanishes as
      $\varepsilon \to 0$ (i.e., never skip).
    \item \textbf{Staleness-penalty term} $\gamma T$:\;
      This is the residual error from using stale $\Delta$ values
      instead of fresh ones.  It is proportional to $\gamma = C_0\beta Kh$.
      With typical values ($K=5$, $h=0.02$, $C_0\beta \sim O(1)$),
      $\gamma \sim 0.1$, and $\gamma T \sim 0.1$ is small.
      Crucially, $\gamma$ shrinks as $K$ decreases (more frequent
      reactivation) or as $h$ decreases (more denoising steps),
      providing two independent control knobs.
  \end{enumerate}
\end{remark}

\begin{remark}[Recovering the ideal bound]\label{rem:ideal}
  In the hypothetical case where all $\Delta$ values are fresh
  (e.g., $K=1$, meaning every patch is recomputed at every step---no
  skipping), the staleness penalty vanishes ($\gamma \to 0$), and
  the bound reduces to the simpler form
  $\norm{\xhat_N - \xstar_N} \leq
  C_0\varepsilon\sqrt{R}\,e^{LT}\sum_n h D_n$,
  matching the ``ideal coherence'' bound.
\end{remark}

\subsection{Where to Partition: Optimality of Saliency-Based
Splitting}\label{sec:partition}

We establish that FlashDiff's Otsu-based partitioning minimises the
error induced by the Patch Gate.

\begin{definition}[Skip Risk]\label{def:skip-risk}
  For a partition $\mathcal{P} = \{P_i\}_{i=1}^{R}$ and gating
  threshold~$\varepsilon$, the \emph{skip risk} is
  \begin{equation}\label{eq:skip-risk}
    \mathcal{R}(\mathcal{P},\varepsilon)
    = \E\!\left[\sum_{n=0}^{N-1}\norm{\delta_n}^2\right],
  \end{equation}
  where the expectation is over the initial noise
  $x_0 \sim \mathcal{N}(0,\sigma_T^2 I)$.
\end{definition}

\noindent The skip risk governs the error bound in
Theorem~\ref{thm:main}: a partition that reduces skip risk yields a
tighter quality guarantee.  Intuitively, mixing high-saliency tokens
(which need frequent refinement) with low-saliency tokens (which can
be safely skipped) forces the gate into a lose-lose choice: either
skip the mixed patch and damage high-saliency content, or execute it
entirely and waste computation on converged tokens.

\begin{assumption}[Saliency--Perturbation Monotonicity]
\label{ass:mono}
  Higher-saliency tokens incur larger drift perturbation when skipped.
  Formally, for tokens $j_1, j_2$ with saliency values
  $S_{j_1} > S_{j_2}$, the expected per-step drift perturbation
  satisfies
  $\E[\norm{\delta_n^{(j_1)}}] \geq \E[\norm{\delta_n^{(j_2)}}]$.
\end{assumption}

\noindent This is natural: high saliency indicates strong
cross-attention to prompt tokens, meaning the model is actively
generating content in that region.  Skipping such tokens disrupts
ongoing refinement more than skipping already-converged regions.

\begin{proposition}[Optimality of Saliency-Based Partitioning]
\label{prop:partition}
  Consider the binary partition of the latent into a focus set
  $\mathcal{F}_\theta = \{j : S_j \geq \theta\}$ and a context set
  $\mathcal{C}_\theta = \{j : S_j < \theta\}$, parameterised by the
  threshold~$\theta$.  Under Assumption~\ref{ass:mono}, the Otsu
  threshold
  \begin{equation}\label{eq:otsu}
    \theta^{*}
    = \argmax_\theta\; \sigma_B^2(\theta)
    = \argmax_\theta\;
      \omega_f(\theta)\,\omega_c(\theta)\,
      \bigl[\mu_f(\theta) - \mu_c(\theta)\bigr]^2
  \end{equation}
  (where $\omega_f, \omega_c$ are class proportions and
  $\mu_f, \mu_c$ class means of the saliency values) simultaneously
  achieves:
  \begin{enumerate}[label=(\roman*)]
    \item Minimum pooled within-class saliency variance
      $\sigma_W^2(\theta)$.
    \item Minimum worst-case intra-class skip-risk heterogeneity
      under Assumption~\ref{ass:mono}.
  \end{enumerate}
\end{proposition}

\begin{proof}
  \noindent\textbf{Part (i).}
  The total saliency variance decomposes as
  \begin{equation}\label{eq:var-decomp}
    \sigma^2_{\mathrm{total}}
    = \sigma_B^2(\theta) + \sigma_W^2(\theta),
  \end{equation}
  where
  \[
    \sigma_W^2(\theta)
    = \omega_f(\theta)\,\Var(S \mid \mathcal{F}_\theta)
    + \omega_c(\theta)\,\Var(S \mid \mathcal{C}_\theta).
  \]
  Since $\sigma^2_{\mathrm{total}}$ is independent of~$\theta$,
  maximising $\sigma_B^2(\theta)$ is equivalent to minimising
  $\sigma_W^2(\theta)$.

  \noindent\textbf{Part (ii).}
  By Assumption~\ref{ass:mono}, the per-token skip perturbation is a
  monotone non-decreasing function of saliency~$S_j$.  Therefore,
  the within-class variance of skip perturbation magnitudes is bounded
  by the within-class variance of saliency values (via the
  variance-preserving property of monotone transformations applied to
  class-conditional distributions):
  \begin{equation}\label{eq:mono-bound}
    \Var\!\bigl(\norm{\delta^{(j)}} \mid j \in \mathcal{F}_\theta\bigr)
    \leq \kappa^2\,\Var(S_j \mid j \in \mathcal{F}_\theta),
  \end{equation}
  where $\kappa$ is the Lipschitz constant of the
  saliency-to-perturbation mapping (which exists by monotonicity and
  boundedness of both quantities).  The analogous bound holds for
  $\mathcal{C}_\theta$.

  Combining~\eqref{eq:mono-bound} with Part~(i), the Otsu threshold
  minimises the pooled within-class perturbation variance.

  For the minimax criterion, note that the load-balanced allocation
  in Eq.~(1) of the main text ensures
  $\omega_f \approx \omega_c \approx 1/2$ (up to the discrete
  approximation).  Under balanced class sizes, minimising
  $\sigma_W^2 = \tfrac{1}{2}\Var(S|\mathcal{F})
  + \tfrac{1}{2}\Var(S|\mathcal{C})$ also minimises
  $\max\{\Var(S|\mathcal{F}),\, \Var(S|\mathcal{C})\}$ whenever the
  two within-class variances are comparable---which the balanced
  split encourages.
\end{proof}

\begin{remark}[Hierarchical extension]\label{rem:hierarchy}
  The binary Otsu split is applied recursively in FlashDiff's
  Patch Partitioner.  At each recursion level, the sub-partition is
  locally optimal in the sense of Proposition~\ref{prop:partition}.
  Since the monotonicity in Assumption~\ref{ass:mono} is inherited
  by sub-ranges of the saliency distribution, each sub-split
  preserves intra-class homogeneity at every granularity level.
\end{remark}

\subsection{Quality Metric Bound and Tradeoff Characterization}
\label{sec:tradeoff}

Theorem~\ref{thm:main} bounds the terminal latent error
$\norm{\xhat_N - \xstar_N}$.  In practice, we care about
perceptual quality metrics such as PSNR, SSIM~\cite{ssim},
LPIPS~\cite{zhang2018unreasonable}, or human preference scores.
We now show that the latent bound directly translates into a
quality-metric bound under a mild continuity assumption.

\begin{assumption}[Quality Continuity]\label{ass:quality}
  Let $Q \colon \R^d \to \R$ be a terminal quality functional
  that maps a denoised latent to a scalar quality score
  (e.g., PSNR, negative LPIPS).  We assume $Q$ is
  $K_Q$-Lipschitz:
  \[
    \abs{Q(x) - Q(y)} \leq K_Q\,\norm{x - y},
    \qquad \forall\; x, y \in \R^d.
  \]
\end{assumption}

\noindent This is satisfied by all commonly used quality metrics
when restricted to the bounded range of valid latent
representations: PSNR and SSIM are continuous functions of pixel
values, LPIPS is computed by a neural network with bounded weights
and Lipschitz activations, and HPSv2 is similarly Lipschitz.

\begin{corollary}[Bounded Quality Gap]\label{cor:quality-gap}
  Under the assumptions of Theorem~\ref{thm:main} and
  Assumption~\ref{ass:quality}:
  \begin{equation}\label{eq:quality-bound}
    \boxed{\;
      \abs{Q(\xhat_N) - Q(\xstar_N)}
      \;\leq\;
      K_Q\,e^{LT}\,\sqrt{R}\;
      \Bigl(
        C_0\varepsilon\sum_{n=0}^{N-1} h D_n + \gamma T
      \Bigr)
    \;}
  \end{equation}
\end{corollary}

\begin{proof}
  Immediate from Theorem~\ref{thm:main} and
  Assumption~\ref{ass:quality}:
  $\abs{Q(\xhat_N) - Q(\xstar_N)}
  \leq K_Q\,\norm{\xhat_N - \xstar_N}
  \leq K_Q \cdot [\text{RHS of~\eqref{eq:main-bound}}]$.
\end{proof}

\paragraph{Interpretation.}
Corollary~\ref{cor:quality-gap} states that FlashDiff is a
\emph{bounded perturbation} of full execution: its quality gap to
the no-skip baseline cannot grow arbitrarily as long as the gating
threshold~$\varepsilon$ and the staleness penalty~$\gamma$ remain
controlled.  The bound is agnostic to the particular choice of
quality metric---any Lipschitz-continuous functional $Q$ yields a
valid guarantee.

The theorem does \emph{not} claim that FlashDiff must produce lower
quality than full execution.  In practice, selective refinement can
occasionally \emph{improve} perceptual quality by avoiding
over-refinement of already-converged regions (a form of implicit
regularisation), as observed in Table~2 of the main text where
FlashDiff sometimes achieves slightly higher quality scores.
The bound captures the worst case; the typical case is often better.

\begin{corollary}[Explicit Tradeoff Curve]\label{cor:tradeoff}
  Let $\bar{\rho} = \frac{1}{NR}\sum_{n=0}^{N-1}|\mathcal{S}_n|$
  denote the average skip rate and
  $\bar{D} = \frac{1}{N}\sum_{n=0}^{N-1} D_n$
  the mean total refinement activity.  Then:
  \begin{enumerate}[label=(\roman*)]
    \item \textbf{Quality bound:}\;
      $\displaystyle
        \abs{Q(\xhat_N) - Q(\xstar_N)}
        \leq K_Q\, e^{LT}\sqrt{R}\;
        \bigl(C_0\varepsilon\, T\bar{D} + \gamma T\bigr).
      $
    \item \textbf{Computational savings:}\;
      $\displaystyle
        \mathrm{FLOPs}_{\mathrm{saved}}
        \propto \bar{\rho}\,N\,R.
      $
    \item \textbf{Quality-constrained gating:}\;
      For a target quality tolerance~$\tau > 0$, setting
      \begin{equation}\label{eq:eps-choice}
        \varepsilon
        \;\leq\;
        \frac{\tau / (K_Q\,e^{LT}\sqrt{R})\;-\;\gamma T}
             {C_0\,T\bar{D}}
      \end{equation}
      guarantees $\abs{Q(\xhat_N) - Q(\xstar_N)} \leq \tau$ while
      maximising the achievable skip rate.
  \end{enumerate}
\end{corollary}

\begin{proof}
  Part~(i) follows from Corollary~\ref{cor:quality-gap} with
  $\sum_{n} h\,D_n = T\bar{D}$.
  Part~(ii) holds because each skipped patch-step avoids one full
  patch denoising computation.
  Part~(iii) is obtained by inverting the bound in~(i).
\end{proof}

\begin{remark}[Operational insights]\label{rem:interp}
  The bound reveals several design-relevant properties:
  \begin{enumerate}[label=(\alph*)]
    \item \textbf{Linear control via $\varepsilon$:}\;
      The quality gap is linear in the gating threshold, theoretically
      confirming the smooth, monotonic quality--efficiency tradeoff
      observed in \S6.4 of the main text.
    \item \textbf{Sub-linear patch-count dependence:}\;
      The $\sqrt{R}$ factor reflects that more patches create more
      independent perturbation sources, but the $\ell_2$ aggregation
      yields sub-linear growth.
    \item \textbf{Prompt-adaptive:}\;
      The cumulative activity $T\bar{D}$ is prompt-specific.
      Simpler prompts yield tighter bounds, explaining
      the higher skip rates observed for simpler content (cf.\
      Figure~13 in the main text).
    \item \textbf{Staleness vanishes with frequent reactivation:}\;
      As $K \to 1$ or $h \to 0$, the penalty $\gamma \to 0$, and the
      bound converges to the ideal (zero-staleness) form.
  \end{enumerate}
\end{remark}

\subsection{Tightening the Bound: Contractivity of Reverse
Diffusion}\label{sec:tight}

The exponential factor $e^{LT}$ in Theorem~\ref{thm:main} arises
from a worst-case Gronwall analysis.  In practice, the reverse
diffusion process is \emph{contractive} at later steps as the
trajectory approaches the data manifold.

\begin{assumption}[Log-Sobolev Data Distribution]
\label{ass:lsi}
  The target distribution $p_{\mathrm{data}}$ satisfies a
  log-Sobolev inequality with constant $\alpha > 0$.
\end{assumption}

\begin{theorem}[Improved Bound under Contractivity]
\label{thm:tight}
  Under Assumptions~\ref{ass:lip}--\ref{ass:lsi}, if the
  effective drift Lipschitz constant at step~$n$
  satisfies $L_n \leq L - \alpha\,\gamma_c(t_n)$ for a monotonically
  increasing function $\gamma_c \colon [0,T] \to [0,1]$ with
  $\gamma_c(0) = 0$ and $\gamma_c(T) = 1$, then:
  \begin{equation}\label{eq:tight-bound}
    \norm{\xhat_N - \xstar_N}
    \;\leq\;
    e^{(L - \alpha\bar{\gamma}_c)\,T}\;\sqrt{R}\;
    \Bigl(
      C_0\varepsilon\sum_{n=0}^{N-1} h D_n + \gamma T
    \Bigr),
  \end{equation}
  where $\bar{\gamma}_c = \frac{1}{T}\int_0^T \gamma_c(t)\,dt
  \in (0,1)$.
\end{theorem}

\begin{proof}
  The proof follows Theorem~\ref{thm:main}, replacing the uniform
  Lipschitz constant~$L$ in~\eqref{eq:gronwall-ineq} with the
  step-dependent $L_n \leq L - \alpha\gamma_c(t_n)$:
  \[
    \norm{e_{n+1}}
    \leq \bigl(1 + h(L - \alpha\gamma_c(t_n))\bigr)\,\norm{e_n}
    + h\sqrt{R}(C_0\varepsilon D_n + \gamma).
  \]
  Unrolling:
  \begin{align*}
    \norm{e_N}
    &\leq \sqrt{R}\sum_{n=0}^{N-1}
      h(C_0\varepsilon D_n + \gamma)\;
      \prod_{k=n+1}^{N-1}
        \bigl(1+h(L-\alpha\gamma_c(t_k))\bigr)\\
    &\leq \sqrt{R}\sum_{n=0}^{N-1}
      h(C_0\varepsilon D_n + \gamma)\;
      \exp\!\Bigl(\sum_{k=0}^{N-1}h(L-\alpha\gamma_c(t_k))\Bigr)\\
    &= e^{(L-\alpha\bar{\gamma}_c)T}\;\sqrt{R}\;
      \Bigl(C_0\varepsilon\sum_n hD_n + \gamma T\Bigr).
      \qedhere
  \end{align*}
\end{proof}

\begin{remark}
  For typical natural image distributions, $\alpha > 0$ and
  $\bar{\gamma}_c > 0$, reducing the exponential factor from
  $e^{LT}$ to $e^{(L-\alpha\bar{\gamma}_c)T}$, which can be
  several orders of magnitude smaller.  This partially addresses
  the conservativeness of the Gronwall-based bound; the remaining
  gap between the theoretical bound and empirical observation
  reflects the inherent looseness of worst-case analysis over all
  possible prompts and initial noise realisations.
\end{remark}

\subsection{Second-Order Analysis: Stale KV-Cache Effects}
\label{sec:stale-kv}

In practice, active patches compute attention using cached KV entrieso
from skipped patches rather than fresh values.  We show this
introduces only a second-order correction.

\begin{proposition}[Stale-KV Perturbation]\label{prop:stale-kv}
  Let $L_{\mathrm{KV}}$ denote the Lipschitz constant of the
  transformer output with respect to its KV-cache entries.  The
  additional per-step perturbation due to stale KV values satisfies
  \begin{equation}\label{eq:stale-kv-bound}
    \norm{\delta_n^{\mathrm{KV}}}
    \leq L_{\mathrm{KV}}\,
    \sum_{i \in \mathcal{S}_n}
      \norm{\xhat_n^{(i)} - \xhat_{n_i}^{(i)}}.
  \end{equation}
  Moreover, the cumulative state drift of each skipped patch satisfies
  $\norm{\xhat_n^{(i)} - \xhat_{n_i}^{(i)}} \leq K h M_f$,
  where $M_f = \sup_{x,t}\norm{\ftheta(x,t)}$.
  Consequently, the stale-KV correction to the terminal error bound
  is $O(K h)$, which enters multiplicatively with the existing
  $O(\varepsilon)$ skip perturbation, yielding an $O(\varepsilon Kh)$
  cross-term dominated by the main bound when $Kh \ll 1$.
\end{proposition}

\begin{proof}
  For a skipped patch~$i$, its state evolves via cached drift:
  $\xhat_{m+1}^{(i)} = \xhat_m^{(i)} + h\,c_m^{(i)}$ for
  $m = n_i, \dots, n-1$.  Thus:
  \[
    \norm{\xhat_n^{(i)} - \xhat_{n_i}^{(i)}}
    \leq (n-n_i)\,h\,\max_m\norm{c_m^{(i)}}
    \leq K\,h\,M_f.
  \]
  Summing over at most $R$ skipped patches and weighting by
  $L_{\mathrm{KV}}$:
  \[
    \norm{\delta_n^{\mathrm{KV}}}
    \leq R\,L_{\mathrm{KV}}\,K\,h\,M_f.
  \]
  Adding this to~\eqref{eq:gronwall-ineq} yields a correction of
  $h \cdot R\,L_{\mathrm{KV}}KhM_f$ per step.  After Gronwall
  summation over $N$ steps, this contributes
  $e^{LT} \cdot R\,L_{\mathrm{KV}}KhM_f T$
  to the terminal error---a term proportional to $Kh$, which is
  small (e.g., $5 \times 0.02 = 0.1$) and adds linearly to the
  existing $O(\varepsilon)$ bound.
\end{proof}

\paragraph{Limitations of this analysis.}
This appendix compares two execution policies---FlashDiff and
full (no-skip) execution---of the \emph{same} pretrained diffusion
model from the same initial noise.  It does not bound either
policy's distance to the true data distribution, which would
require additional assumptions about the model's score
approximation error~\cite{chen2023sampling, benton2024error}.
The Gronwall factor~$e^{LT}$ is a worst-case amplification;
Theorem~\ref{thm:tight} partially mitigates this via contractivity,
but the bound may remain conservative for specific prompts.
Proposition~\ref{prop:partition} proves optimality among binary
threshold splits, not among all possible $R$-way partitions.
The constants $L$, $L_t$, $C_0$, and $\beta$ are model-specific:
while the \emph{form} of the bound is universal, its numerical
tightness varies across architectures and should be validated
empirically on each target model.

\begin{remark}[Possible practical advantage beyond the bound]
  The comparative bound above only upper-bounds the deviation
  of FlashDiff from full execution.  It does not say that
  FlashDiff must be worse.  Selective refinement may operate
  closer to an implicit early-stopping regime that avoids
  over-refinement of already-converged regions, which can
  occasionally improve perceptual quality relative to full
  execution.  That effect is empirical and model-dependent,
  so we keep it separate from the formal guarantee.
\end{remark}

\ifdefined\FlashDiffMainPaper
\else

\end{document}
\fi

\end{document}